\documentclass[a4paper]{article}

\usepackage{amsmath,amssymb,amsthm}
\usepackage{geometry}
\usepackage{mathptmx}

\usepackage{amsmath}
\usepackage{hyperref}
\usepackage[english]{babel}

\usepackage{algorithm}
\usepackage{algpseudocode}

\usepackage{amsfonts}

\usepackage{soul} 
\usepackage{cases}
\usepackage{multirow}

\usepackage[most]{tcolorbox} 
\usepackage{pifont}
\usepackage{booktabs}
\usepackage{xspace}

\usepackage{mathtools}

\usepackage{subfigure}

\newcommand{\ie}{i.e.,\xspace}
\newcommand{\eg}{e.g.,\xspace}

\newcommand{\simrel}[0]{\preccurlyeq}

\newcommand{\fct}[1]{\set{#1}} 
\newcommand{\set}[1]{\ensuremath{\mathsf{#1}}}
\newcommand{\fctnew}[1]{\ensuremath{\mathsf{#1}}}

\renewcommand{\max}[0]{\fctnew{max}}
\renewcommand{\min}[0]{\fctnew{min}}
\newcommand{\avg}[0]{\fctnew{avg}}

\newcommand{\defeq}[0]{=}

\newcommand{\Succ}[1]{\set{Succ}\left(#1\right)}
\newcommand{\Reach}[1]{\set{Reach}\left(#1\right)}
\newcommand{\Safe}{\ensuremath{\text{\textit{Safe}}}} 
\newcommand{\Unsafe}{\ensuremath{\text{\textit{Unsafe}}}}

\newcommand{\Max}[1]{\ensuremath{\max^{\simrel}\!\left(#1\right)}}
\newcommand{\dc}[1]{\ensuremath{\downarrow^{\simrel}\!\!\left(#1\right)}}
\newcommand{\uc}[1]{\ensuremath{\uparrow^{\simrel}\!\!\left(#1\right)}}
\newcommand{\Rtilde}{\ensuremath{\widetilde{R}}}
\newcommand{\Ntilde}{\ensuremath{\widetilde{N}}}

\newcommand{\nat}{\fctnew{nat}}
\newcommand{\rct}{\fctnew{rct}}
\newcommand{\crit}{\fctnew{cri}}
\newcommand{\systemstate}[1][]{\ensuremath{\langle \rct_{S^{#1}}, \nat_{S^{#1}}, \crit_{S^{#1}} \rangle}}
\newcommand{\ttd}{\fctnew{ttd}}
\newcommand{\ttvd}{\fctnew{ttvd}}

\newcommand{\dbf}{\fctnew{dbf}}
\newcommand{\df}{\fctnew{df}}
\newcommand{\nj}{\fctnew{nj}}

\newcommand{\States}[1]{\ensuremath{\set{States}\left(#1\right)}}
\newcommand{\Active}[1]{\ensuremath{\set{Active}\left(#1\right)}}
\newcommand{\Eligible}[1]{\ensuremath{\set{Eligible}\left(#1\right)}}
\newcommand{\Completed}[1]{\ensuremath{\set{Completed}\left(#1\right)}}
\newcommand{\DeadlineMiss}[1]{\ensuremath{\set{DeadlineMiss}(#1)}}

\newcommand{\run}{\fctnew{run}}
\newcommand{\ran}{\fctnew{ran}}
\newcommand{\schedule}{\fctnew{sch}}

\newcommand{\hi}{\ensuremath{\fct{HI}}}
\newcommand{\lo}{\ensuremath{\fct{LO}}}

\newcommand{\true}{\ensuremath{\fct{True}}}
\newcommand{\false}{\ensuremath{\fct{False}}}

\newcommand{\edfvd}{\ensuremath{\schedule_{\text{EDF-VD}}}}
\newcommand{\lwlf}{\ensuremath{\schedule_{\text{LWLF}}}}

\newcommand{\srun}{\ensuremath{\prime}} 
\newcommand{\scom}{\ensuremath{-}}
\newcommand{\sreq}{\ensuremath{+}}

\newcommand{\idle}{\ensuremath{\simrel_{idle}}}

\newcommand{\laxity}{\fctnew{laxity}}
\newcommand{\wlaxity}{\fctnew{worstLaxity}}
\newcommand{\sumlax}{\ensuremath{\ell}}

\newcommand{\drop}[1]{}

\newcommand{\sigCmp}[2]{\ensuremath{\fctnew{sigCmp}_{#1}\left(#2\right)}}
\newcommand{\critUp}[2]{\ensuremath{\fctnew{critUp}_{#1}\left(#2\right)}}
\newcommand{\SgTrans}[1]{\ensuremath{\xrightarrow{#1}_{\mathsf{sg}}}}
\newcommand{\CpTrans}[1]{\SgTrans{#1}} 
\newcommand{\RnTrans}[1]{\ensuremath{\xrightarrow{#1}_{\mathsf{rn}}}}
\newcommand{\CtTrans}[1]{\RnTrans{#1}} 
\newcommand{\RlTrans}[1]{\ensuremath{\xrightarrow{#1}_{\mathsf{rl}}}}
\newcommand{\RqTrans}[1]{\RlTrans{#1}} 

\newcommand{\MinD}{\ensuremath{\fctnew{min_d}}}
\newcommand{\MinVD}{\ensuremath{\fctnew{min_{vd}}}}
\newcommand{\MinWL}{\ensuremath{\fctnew{min_{wl}}}}

\newcommand{\etal}{\textit{et al.}}
\newcommand{\sth}{s.t.}

\newcommand{\emphindef}[1]{\emph{\textbf{#1}}}

\definecolor{darkgreen}{RGB}{0,120,0}

\newcommand{\cmark}{\textcolor{darkgreen}{\ding{52}}}
\newcommand{\xmark}{\textcolor{red}{\ding{54}}}

\newcommand{\speedup}{99.998}
\newcommand{\timeoutMins}{15}

\newcounter{oraclecounter}
\newtcolorbox[auto counter]{oraclebox}[3][]{%
    colbacktitle=gray!20, 
    coltitle=black, 
    colback=white, 
    colframe=black, 
    sharp corners,
    title={\textbf{Oracle \thetcbcounter: #2.} #3}, 
    #1,
    fonttitle=\sffamily,
    detach title,
    before upper={\tcbtitle\par\medskip}, 
    colframe=black, colback=white, 
    coltext=black, 
    boxsep=5pt, 
    arc=0pt,outer arc=0pt, 
}

\newtheorem{definition}{Definition}
\newtheorem{proposition}{Proposition}
\newtheorem{theorem}{Theorem}
\newtheorem{lemma}{Lemma}
\newtheorem{corollary}{Corollary}
\newtheorem{oracle}{Oracle}

\addto\extrasenglish{%
}

\begin{document}

\title{Exact schedulability test for sporadic mixed-criticality real-time systems using antichains and oracles}

\author{Picard, Simon\\
\url{snspicard@gmail.com}\\
Ourway, Belgium
\and
Paolillo, Antonio\\
\url{antonio.paolillo@vub.be}\\
Vrije Universiteit Brussel, Belgium
\and
Geeraerts, Gilles\\
\url{gilles.geeraerts@ulb.be}\\
Universit\'e libre de Bruxelles, Belgium
\and
Goossens, Joël\\
\url{joel.goossens@ulb.be}\\
Universit\'e libre de Bruxelles, Belgium
}

\maketitle

\begin{abstract}
This work addresses the problem of exact schedulability assessment in
uniprocessor mixed-criticality real-time systems with sporadic task sets.
We model the problem by means of a finite automaton that has to be
explored in order to check for schedulability.
To mitigate the state explosion problem, we provide a generic
algorithm which is parameterised by several techniques called
oracles and simulation relations.
These techniques leverage results from the scheduling
literature as ``plug-ins'' that make the algorithm more efficient in practice.
Our approach 
achieves up to a \speedup{}\% reduction in the search space
required for exact schedulability testing, making it practical for a range of task sets, up to 8 tasks or maximum periods of 350.
This method enables to challenge the pessimism of an existing schedulability test
and to derive a new dynamic-priority scheduler, demonstrating its good performance.
This is the full version of an RTNS 2024 paper \cite{PPGG24}.
\end{abstract}

\section{Introduction}
\label{sec:introduction}

The industry push for integrating systems of varying criticality levels onto a single platform has driven the real-time community to develop the mixed-criticality system model, hence producing a large body of research results for systems designed to degrade gracefully~\cite{burns2022mixed}.
The dual-criticality model
enables the concurrent execution of high-criticality tasks (\hi{}) alongside low-criticality tasks (\lo{}).
The \hi{} tasks are certified with two estimates of their worst-case execution time
--- an optimistic estimate and a pessimistic upper bound ---
whereas the \lo{} tasks are less critical and can be safely disabled should a \hi{} task exceed its optimistic estimate.

Many results about this model have been produced,
addressing
variations of the model~\cite{ekberg2014generalmcs,baruah2016mixedPeriods}
varying CPU speeds~\cite{
    she2022constrainedDvfsMixcrit,
    she_rtns2021_mixcritSpeedMultiproc},
multi-core and parallel systems~\cite{haohan_ecrts2012_globalSchedMixcrit,baruah_rtss2016_fedDagMixcrit},
and applications in industrial settings~\cite{
    esper2018industrial,
    law_ecrts2019_industrialMixcrit,
    paolillo2017porting}.
Despite this progress, however, one fundamental problem remains open:
given a generic scheduling algorithm,
how to assess the schedulability of a mixed-criticality task set?
%
Agrawal and Baruah proved that determining the schedulability of
independent dual-criticality periodic or sporadic implicit-deadline
tasks is NP-hard in the strong
sense~\cite{agrawalbaruah_ecrts2018_intractableMixcrit}, hinting that
the schedulability assessment of a given scheduling algorithm might be
hard as well. Indeed, to the best of our knowledge, no
exact --- necessary \emph{and} sufficient --- schedulability test has
been produced to date for a general scheduling algorithm, even for a
uniprocessor platform.

The quest for schedulability tests has produced
two main lines of research. On the one hand, the real-time community
has focussed on efficient tests for specific schedulers, but they are
in many cases not exact. For example, EDF-VD~\cite{BaruahBDMSS11}, a
notable scheduling algorithm for mixed-criticality systems only
provides a \emph{sufficient} and notoriously pessimistic test. On the
other hand, inspired by the formal methods community, several works
have proposed to model the possible behaviours of the system by means
of an \emph{automaton}~\cite{bakerbrute} whose states correspond to
all the possible states of the system. In this case, looking for
potential deadline misses can be done by analysing all the states of
the automaton. While such approaches are guaranteed to provide an
exact test \cite{asyaban2018exact,yasmina2020accurate}, they suffer
from the state explosion problem making them impractical for
realistic systems.

In this work, we seek to reconcile both approaches. While our
work is rooted in the automaton-based approach, we show how results
and knowledge from the real-time community can be exploited to make
the traversal of automata states more tractable. We
believe that such hybrid techniques are very rare in the literature,
with some notable exceptions like the works of Asyaban
\cite{asyaban2018exact} (which addresses mixed criticality in the FTP
case) and Ranjha~\cite{9804591}. However, those
papers develop \textit{ad hoc} techniques (in the cases of FTP and FJP
scheduling), while we strive for more
generality.

Our contributions are thus as follows.  (1) We provide an automaton
model for systems of dual-criticality sporadic tasks on a uniprocessor
platform where the scheduler is left as a parameter. Hence, (2) we
obtain an \emph{\textbf{exact}} schedulability test for any given
scheduling algorithm that is deterministic and
memoryless. 
(3) This test consists in exploring the states of the automaton and is
parameterised by optimisations to leverage knowledge built in the
real-time community. More precisely, sufficient and necessary tests
can be exploited as \emph{oracles} that tell the algorithm to avoid
exploring some states. We also exploit a \emph{simulation relation}
between states to further prune the state space in the spirit of the
antichain approach of the formal method community
\cite{DDHR06,DR10,DR09,geeraerts2013multiprocessor}.  (4) We show how
to exploit this generic framework in our setting of mixed criticality
by defining proper oracles and a simulation relation. On this basis,
we evaluate empirically our approach.  (5) Experiments on task sets
generated randomly according to the model specification show promising
results: the efficiency of our approach enables to reduce the state
search space by up to \speedup{}\%.
The genericity of our approach also allows us (6) to evaluate the
EDF-VD scheduling algorithm and its associated test, showing that
the test is pessimistic
w.r.t. its actual scheduling capabilities.
(7) For illustration, we study and present a new dynamic-priority scheduler, LWLF.
We demonstrate its excellent performance, thanks to its anticipation of mode change impact when operating in $\lo$ mode.
\autoref{table:comparison-related} graphically compares our approach with the closest related work.

\newcommand{\singlecrit}{\textcolor{orange}{Single}}
\newcommand{\dualcrit}{\textbf{\textcolor{darkgreen}{Dual}}}
\newcommand{\ftpclass}{\textcolor{orange}{FTP}}
\newcommand{\anyclass}{\textbf{\textcolor{darkgreen}{Any}}}

\begin{table}[t]
  \centering
    \caption{%
    Comparing with the related work.
    }
    \label{table:comparison-related}
  
\begin{tabular}{lcccc}
\toprule
& \cite{bakerbrute}
& \cite{asyaban2018exact}
& \cite{lindstrom2011faster}
& Us \\
\midrule
Criticality         & \singlecrit{} & \dualcrit{}   & \singlecrit{} & \dualcrit{}   \\
Scheduling classes  & \anyclass{}   & \ftpclass{}   & \anyclass{}   & \anyclass{}   \\
Pruning rules       & \xmark        & \cmark        & \xmark        & \xmark        \\
Antichains          & \xmark        & \xmark        & \cmark        & \cmark        \\
Oracles             & \xmark        & \xmark        & \xmark        & \cmark        \\
Multi-processor     & \cmark        & \xmark        & \cmark        & \xmark        \\
\bottomrule
\end{tabular}
\end{table}


\section{Problem definition}
\label{sec:problemdef}

We consider a mixed-criticality sporadic task set $\tau =
\{\tau_1, \tau_2, \ldots,\tau_n\}$ with two levels of criticality, also referred to as a dual-criticality~\cite{burns2022mixed}, to be scheduled on a uniprocessor platform.
A dual-criticality sporadic task
$\tau_i \defeq \langle \langle C_i(\lo), C_i(\hi) \rangle, D_i, T_i, L_i \rangle$
is characterised by a \textit{minimum
interarrival time} $T_i > 0$, a \textit{relative deadline} $D_i > 0$, a \textit{criticality
level} $L_i \in \{\lo,\hi\}$ with $\hi > \lo$ and the \emph{worst-case execution time} (WCET) tuple $\langle C_i(\lo),
C_i(\hi)\rangle$.
Tasks will never execute for more than $C_i(L_i)$.
Time is assumed to be discrete, ergo $\forall i: T_i, D_i, C_i(\lo), C_i(\hi) \in
\mathbb{N} \setminus \{0\}$.
I.e., all timing parameters are strictly-positive integers.
It is assumed that $C_i(\lo) \leq C_i(\hi)$ when $L_i = \hi$ and $C_i(\hi) = C_i(\lo)$  when $L_i = \lo$.
A dual-criticality sporadic task $\tau_i$ releases an infinite number of jobs, with each job release being separated by at least $T_i$ units of time.
The absolute deadline of jobs are set $D_i$ units of time after their release, and jobs must signal completion before their absolute deadline.
We assume jobs are independent, as formulated by Vestal~\cite{vestal2007preemptive}.
At each clock-tick, the executing job can signal its completion.
If a job did not signal completion after exhausting its $C_i(\lo)$, then a mode change is triggered from $\lo{}$ to $\hi$:
jobs of $\lo{}$ tasks are discarded
and $\lo{}$ tasks are not allowed to release jobs any more;
active jobs of $\hi{}$ tasks receive an additional budget of $C_i(\hi) - C_i(\lo)$, and future jobs of $\hi{}$ tasks will receive a budget of $C_i(\hi)$.
%
%
As a running example in the work,
we define the task set $\tau^{a} = \lbrace \tau_1, \tau_2 \rbrace$
with $\tau_1 = \langle \langle 1, 2 \rangle, 2, 2, \hi \rangle$
and  $\tau_2 = \langle \langle 1, 1 \rangle, 2, 2, \lo \rangle$.

It is imperative,
when validating a dual-criticality task sets,
to first perform \emph{``due diligence''}
by ensuring the corresponding two following single-criticality task sets are schedulable:
\textit{(i)} the task set comprising only $\hi$ tasks, where
$\forall \tau_i: C_i = C_i(\hi)$,
and \textit{(ii)} the task set including both $\hi$ and $\lo$ tasks, where
$\forall \tau_i: C_i = C_i(\lo)$.

We aim to establish an exact schedulability test (necessary and sufficient) for any dual-criticality sporadic task set
$\tau$ that tells us whether the set is schedulable
--- \ie no job released by the tasks misses a deadline ---
with a given deterministic and preemptive scheduling algorithm with dynamic priorities.
We support both \emph{implicit} ($\forall \tau_i : D_i = T_i$) and \emph{constrained} ($\forall \tau_i : D_i \leq T_i$) deadlines.
%
%
The studied problem is difficult, notably due to two sources of non-determinism, the first relating to \textbf{the sporadic model} (we do not know when jobs are released) and the second to \textbf{mode change} (we do not know when it will occur).

\textbf{Notations.}
Assuming $\alpha, \beta \in \lbrace \lo, \hi \rbrace$:
the $\alpha$-utilisation of the task $\tau_i \in \tau: U^\alpha(\tau_i) = C_i(\alpha)/T_i$,
the $\alpha$-utilisation of the task set $\tau: U^\alpha(\tau) = \sum_{\tau_i \mid L_i \geq \alpha} U^\alpha(\tau_i)$,
the $\alpha$-utilisation of the tasks of criticality $\beta$ in the task set $\tau: U^\alpha_\beta(\tau) =  \sum_{\tau_i \mid L_i = \beta} U^\alpha(\tau_i)$,
and the average utilisation of the task set $\tau: U^{\avg}(\tau) = \frac{U^\lo(\tau)+U^\hi(\tau)}{2}$.

\section{Automaton based semantic}
\label{sec:automaton-semantic}

To obtain a formal definition of the semantics of the system and of the
schedulability problem, we develop an automaton-based formalism inspired from Baker and
Cirinei~\cite{bakerbrute}.
An automaton is a graph whose nodes are called \emph{states} and whose
(directed) edges are called \emph{transitions}.
States model the states of the system.
Changes from a state to another, triggered by a task, a job, or the scheduler,
are modelled by the transitions.
Hence, each path in the automaton is a possible execution of the system, and we will thus
look for executions that reach states where deadlines are missed (the so-called
\emph{failure states}) to check whether the system is schedulable or not.
We use the following formal definition of automaton:

\begin{definition}
  An \emphindef{automaton} is a tuple $A = \langle V,E,v_0,F \rangle$ where
  $V$ is a set of \textit{states},
  $E \subseteq V \times V$ is a set of \textit{transitions} between states,
  $v_0 \in V$ is the \textit{initial state}
  and $F \subseteq V$ is a set of \textit{failure states}.
  An automaton is finite iff $V$ is a finite set.

\end{definition}

The problem we consider on automata is the problem of \textit{safety} w.r.t. a designated set of
\emph{failure} states $F$ that need to be avoided.
A \emph{path} in a finite automaton $A = \langle V, E, v_0, F\rangle$ is a finite sequence
of states $v_1, v_2, \ldots, v_k$ such that $\forall\, 1 \leq j < k : (v_j, v_{j+1})
\in E$.
For a subset of states $V' \subseteq V$, if there exists a path $v_1,v_2,\ldots, v_k$ in
$A$ such that $v_k \in V'$, we say that $v_1$ can reach $V'$.
We denote the set of states that can be reached from a state $v \in V$ by $\Reach{v}$.
Then, the \textit{safety problem} asks, given an automaton $A$, whether the initial state
$v_0$ \emph{cannot} reach the set of failure states $F$ --- \ie there is no path from $v_0$
to any state $v \in F$ in the automaton, denoted by $\Reach{v_0} \cap F = \emptyset$.
If this is the case, we say that $A$ is \emph{safe}, otherwise (when $v_0$ can reach $F$)
we say it is \emph{unsafe}.

Let $\tau = \{\tau_1, \tau_2, \ldots,\tau_n\}$ be a task set as defined in~\autoref{sec:problemdef}.
We model the behaviour of $\tau$
by means of an automaton $A$, and we reduce the schedulability problem of
$\tau$ to an instance of the safety problem in $A$.
A state captures the following
dynamic (or ``run-time'') information about each task $\tau_i$:
\textit{(i)} the \textit{earliest next arrival time} $\nat(\tau_i)$ relative to the current instant, and
\textit{(ii)} the \textit{worst-case remaining execution time} $\rct(\tau_i)$ of the current job of $\tau_i$ for the current level of criticality.
In addition, we need to remember the current global \textit{criticality level} (or mode)
$\crit \in \lbrace \lo, \hi \rbrace$ of the system.
Hence, each system state will be a tuple of the form $\systemstate$:

\begin{definition}[System states]\label{systemstate}
  Let $\tau = \{\tau_1, \tau_2, \ldots, \tau_n\}$ be a mixed-criticality sporadic task system.
  Let $T_{\max} \defeq \max_{i} T_i$ and $C_{\max} \defeq \max_{i} C_i(L_i)$.
  A \emphindef{system state} of $\tau$ is a tuple $S = \systemstate$ where:
  $\nat_S: \tau \mapsto \{0,1,\ldots,T_{\max}\}$ associates each task $\tau_i$ to the minimal delay
  $\nat_S(\tau_i)$ that must elapse before the next job of the task can be released;
  $\rct_S: \tau \mapsto \{0,1,\ldots,C_{\max}\}$ associates each task $\tau_i$ to its maximal 
  remaining execution time $\rct_S(\tau_i)$ for the current criticality of the system; and $\crit_S \in \{\lo,\hi\}$ is the current
  criticality of the system.
  We denote by $\States{\tau}$ the set of all system states of $\tau$.
\end{definition}

Intuitively, each state contains the current run-time information of the systems at a particular instant,
\ie how far we are in the execution of the active jobs (with the $\rct_S(\tau)$ values)
and how close we are to the next releases of jobs (with the $\nat_S(\tau)$ values).
From these definitions, we derive other useful definitions in the following.
Notice that the symbols $\wedge$, $\vee$ and $\neg$ respectively denote the \textit{conjunction}, \textit{inclusive disjunction} and \textit{negation} logical operators.

\begin{definition}[Time to deadline]\label{ttd}
  Let $\ttd_S(\tau_i) \defeq \nat_S(\tau_i)-(T_i-D_i)$ be the \emphindef{time to deadline},
  the time remaining before the absolute deadline of the last submitted job~\cite{bakerbrute}
  of $\tau_i \in \tau$ in state $S$.
  Note that when deadlines are implicit, we have $\ttd_S(\tau_i) = \nat_S(\tau_i)$.
\end{definition}

\begin{definition}[Active tasks]\label{act}
  A task $\tau_i$ is \emphindef{active} in a state $S$ iff it currently has a
  job that is not completed in $S$.
  The set of active tasks in $S$ is $\Active{S} \defeq \{\tau_i \mid \rct_S(\tau_i) > 0 \}$.
\end{definition}

\begin{definition}[Eligible task]\label{eligible}
  A task $\tau_i$ is \emphindef{eligible} in the state $S=\systemstate$ iff it can release a
  job in this state
  --- \ie the task does not currently have an active job,
  the last job was submitted at least $T_i$ time units ago
  and its criticality is greater than or equal to the state's.
  The set of eligible tasks in $S$ is:
  $\Eligible{S} \defeq \{\tau_i \mid \rct_S(\tau_i) =  \nat_S(\tau_i) = 0 \wedge L_i \geq \crit_S\}$.
\end{definition}

\begin{definition}[Implicitly completed task]\label{completed}
  A task $\tau_i$ is \emphindef{implicitly completed} in the state $S=\systemstate$ iff its
  latest job has been executed for its maximal execution time.
  The set of implicitly completed tasks in $S$ is
  $\Completed{S} \defeq \{\tau_i \mid \rct_S(\tau_i) = 0 \wedge C_i(\crit_S)=C_i(L_i)\}$.
\end{definition}

Note that $C_i(\crit_S)=C_i(L_i)$ represents the condition ``its latest job has been
executed for its maximal execution time''.
The condition $\crit_S=L_i$ is not appropriate because it is possible for a $\hi$-critical
job to have the same execution time in $\lo$ and $\hi$.
Therefore, when $C_i(\lo) = C_i(\hi)$,
$\rct_S(\tau_i) = 0$ and $\crit_S=\lo$
then $\tau_i$ is implicitly completed because
$C_i(\crit_S)=C_i(\lo)=C_i(L_i)=C_i(\hi)$ even if $\crit_S=\lo\neq L_i=\hi$.

\begin{definition}[Deadline-miss state]\label{fail}
  A state $S$ is a \emphindef{deadline-miss state} if at least one task's job reached its
  deadline without executing all its maximal execution time.
  The set of deadline-miss states on $\tau$ is
  $\DeadlineMiss{\tau} \defeq \{S \mid  \exists \tau_i \in\tau : \rct_S(\tau_i) > 0 \wedge \ttd_S(\tau_i) \leq 0 \}$.
\end{definition}

\begin{definition}[Scheduler]\label{run}
  A uniprocessor \emphindef{memoryless scheduler} for $\tau$ is a function $\schedule: \States{\tau}
  \mapsto \tau \cup \{\bot\}$ s.t. $\schedule(S) \in \Active{S}$ or $\schedule(S) = \bot$ when no task is to be scheduled.
  Moreover, we say that the scheduler $\schedule$ is \emphindef{deterministic}, iff for all $S_1, S_2 \in \States{\tau}$ s.t.
      $\Active{S_1} = \Active{S_2}$ and $\crit_{S_1} = \crit_{S_2}$, for all $\tau_i \in
      \Active{S_1}$ : $\nat_{S_1}(\tau_i) = \nat_{S_2}(\tau_i) \wedge \rct_{S_1}(\tau_i) =
      \rct_{S_2}(\tau_i)$ implies $\schedule(S_1) = \schedule(S_2)$.
\end{definition}

%
A \emph{memoryless scheduler} is a scheduler that makes decisions based only on the current state, and not on the history of previous states.
The \emph{deterministic} property means that 
the scheduler always takes the same decision given the same criticality level and the same active task characteristics.
It implies that the decisions of a deterministic scheduler do not involve
any randomness in the job selection and are unaffected by the inactive tasks.
In this work, we will only consider deterministic memoryless schedulers.
As an example, we define the EDF-VD scheduler~\cite{BaruahBDMSS11} within our framework:
\begin{definition}[EDF-VD scheduler]\label{edfvd}
  With 
  $\lambda \defeq \frac{U_\hi^\lo(\tau)}{1-U_\lo^\lo(\tau)}$
  a discount
  factor, let $\ttvd_S(\tau_i)$ be the time remaining before the
  virtual deadline of the last submitted job of $\tau_i \in \Active{S}$ in state $S$
  defined as follows:
    \begin{numcases} {\ttvd_S(\tau_i)=}
        \nat_S(\tau_i)-(T_i-D_i \cdot \lambda)    & \textnormal{if} $L_i = \hi$ \nonumber \\
        \ttd_S(\tau_i)                            & \textnormal{otherwise.} \nonumber
    \end{numcases}

  Further, we let $\MinD(S)$ be the task $\tau_i\in\Active{S}$ which has the minimal
  deadline in $S$.
  That is, $\tau_i$ is s.t. for all $\tau_k\in \Active{S}\setminus\{\tau_i\}$:
  $\ttd_S(\tau_k) > \ttd_S(\tau_i)$ or $\ttd_S(\tau_k) = \ttd_S(\tau_i) \wedge k > i$.
  We define $\MinVD(S)$ similarly, using virtual deadlines instead of deadlines
  (substituting $\ttvd_S$ for $\ttd_S$ in the definition).
  Then, for all states $S$, we let: $\edfvd(S)$
  \begin{numcases}
    {=}
      \bot &  \textnormal{if} $\Active{S}=\emptyset$ \nonumber \\
      \MinD(S)  &   \textnormal{else if} $\crit_S = \hi \vee U_\lo^\lo(\tau) + U_\hi^\hi(\tau) \leq 1$ \nonumber \\
      \MinVD(S) &  \textnormal{otherwise.} \nonumber
  \end{numcases}
\end{definition}

In the two first cases, the scheduler is behaving exactly as EDF.
The original work of EDF-VD~\cite{BaruahBDMSS11} assumes implicit deadlines (\ie $\forall \tau_i: D_i = T_i)$.
Notice we cannot have $\lambda > 1$
as it would mean $U^{\lo} > 1$,
the task set not being schedulable in $\lo$ mode, contradicting the requisites outlined in Section 2 (\emph{``due diligence''}).

Thanks to these notions, we can define the transitions of the automaton.
Those transitions must embed all the modifications happening to the system state within one clock-tick.
Each of those modifications can be seen as an intermediary transition between two system states, and they happen in the following order:
(1) \textit{Release transitions} model the release of jobs by sporadic tasks at a given instant in time,
(2) \textit{Run transitions} model the elapsing of one time unit, and running the job selected by the scheduler, if any
and (3) \textit{Signal transitions} model a job signalling, or not, completion, and a potential resulting mode change.

In the automaton, an \emph{actual transition} will exist only if those three intermediary transitions happen one after the other.
We now formally define how each kind of intermediary transitions alter the state.
We start by defining \emph{release transitions}.
Let $S$ be a state in $\States{\tau}$.
Intuitively, when the system is in state $S$,
a task $\tau_i$ releasing a new job has the effect to update $S$ by
setting $\nat(\tau_i)$ to $T_i$ and $\rct(\tau_i)$ to $C_i(\crit_S)$.
Formally:

\begin{definition}[Release transition]\label{treq}
  Let $S = \systemstate \in \States{\tau}$ be a system state and $\tau^\sreq \subseteq
  \Eligible{S}$ be a set of tasks that are eligible to release a new job in the system.
  Then, we say that $S^\sreq = \systemstate[\sreq]$ is a \emphindef{$\tau^\sreq$-release
  successor} of $S$, denoted $S \RqTrans{\tau^\sreq} S^\sreq$ iff:

  \begin{enumerate}
    \item $\forall \tau_i \in \tau^\sreq : \nat_{S^\sreq}(\tau_i) = T_i$ and
      $\rct_{S^\sreq}(\tau_i)=C_i(\crit_S)$
    \item $\forall \tau_i \notin \tau^\sreq : \nat_{S^\sreq}(\tau_i) = \nat_S(\tau_i)$ and
      $\rct_{S^\sreq}(\tau_i)=\rct_S(\tau_i)$
    \item $\crit_{S^\sreq} = \crit_S$.
  \end{enumerate}
\end{definition}

Notice that we allow $\tau^\sreq = \emptyset$,
that is, no task releases a new job in the system.
Also note that no time elapsed in that transition, so no $\nat(\tau_i)$ nor $\rct(\tau_i)$ must be updated.
Furthermore, observe that changing the definition with $\tau^\sreq = \Eligible{S}$ would lead to consider a periodic task model (without offset),
as all tasks that could release a job would then have to release it immediately.
Next, we move to \emph{run transitions}.
Let $S$ be a state in $\States{\tau}$, and let $\run$ be the scheduling decision to apply, \ie either $\run \in \tau$ and $\run$ must be executed or $\run = \bot$ and the processor remains idle.
Then, letting one time unit elapse from $S$ under the $\run$ scheduling decision
leads to decrementing the $\rct$ of the task $\run$ (and only this task) if $\run \in \tau$, and to
decrementing the $\nat$ of all tasks.
Formally:

\begin{definition}[Run transition]\label{texec}
  Let $S = \systemstate \in \States{\tau}$ be a system state and $\run \in \tau \cup \{\bot\}$ be a task to be executed or $\bot$ if no task is to be executed.
  Then, we say that $S' = \systemstate[\srun]$ is a \emphindef{run successor} of $S$
  under $\run$, denoted $S \CtTrans{\run} S^\srun$ iff:
  \begin{enumerate}
    \item For all $\tau_j\in\tau\setminus\{\run\}$: $\rct_{S^\srun}(\tau_j) = \rct_S(\tau_j)$ and
      $\run \neq \bot$ implies $\rct_{S^\srun}(\run) = \rct_S(\run)-1$
    \item $\forall \tau_i \in \tau : \nat_{S^\srun}(\tau_i) = \max\{\nat_S(\tau_i)-1, 0\}$
    \item $\crit_{S^\srun} = \crit_{S}$.
  \end{enumerate}
\end{definition}

Finally, let us define \emph{signal transitions}.
When the system is in state $S \in \States{\tau}$, there are typically two possible scenarios for the task $\tau_r$ which has just been executed: it can signal completion, or not.
We call the combination of this information a \emph{setup}: $S, \tau_r$ and whether $\tau_r$ signals completion.
This non-deterministic behaviour leads to three different outcomes for the state $S$.
Outcome (1) is when $S$ does not change at all, which can result from different
setups.
Either no task has been executed;
or $\tau_r$ must signal completion because it has exhausted all its execution time budget for its worst criticality mode, i.e. $\tau_r$ is \emph{implicitly} completed;
or the task $\tau_r$ does not signal completion without having exhausted all its execution time budget;
or the task $\tau_r$ signals completion and has exhausted all its execution time budget for the current criticality mode, which is not its worst one.
Outcome (2) is when $\tau_r$ signals completion with remaining execution time budget for the current criticality mode, i.e. $\tau_r$ is \emph{explicitly} completed.
Outcome (3) happens when $\tau_r$ does not signal completion,
has exhausted its execution time budget for the current criticality mode but still has execution time budget in the higher ($\hi$) criticality
mode, which triggers a mode change.
When a task $\tau_r$ signals completion explicitly, then the resulting state is identical to $S$ except for the $\rct$ of $\tau_r$, which is set to $0$.
Such state is formalised as $\sigCmp{S}{\tau_r}$ where $\sigCmp{S}{\tau_r}$ is the state $S'$ s.t.
: $\nat_S=\nat_{S'}$; $\crit_S=\crit_{S'}$; $\rct_{S'}(\tau_r)=0$; and, for all
$\tau_i\in\tau\setminus\{\tau_r\}$: $\rct_{S'}(\tau_i)=\rct_S(\tau_i)$.
When a mode change occurs, $\crit$ becomes $\hi$ and all the active
$\hi$-critical tasks see their $\rct$ increase by the difference between their
execution time in levels $\hi$ and $\lo$.
All $\lo$-critical tasks are discarded, having $\rct$ set to $0$.
The state obtained from $S$ when $\tau_r$ triggers a mode change is thus denoted as
$\critUp{S}{\tau_r}$. Hence, $\critUp{S}{\tau_r}$ is formalised as the state $S'$ s.t. :
$\nat_S=\nat_{S'}$; $\crit_{S'}=\hi$; and, for
all $\tau_i \in \tau:\rct_{S'}(\tau_i) $
  \begin{numcases}
    {=}
        \rct_S(\tau_i)+ C_i(\hi)-C_i(\lo) & if $L_i = \hi \wedge \rct_S(\tau_i) > 0$ \nonumber \\
         C_i(\hi)-C_i(\lo) & else if $\tau_i = \tau_r$ \nonumber \\
        0 & otherwise. \nonumber
  \end{numcases}

Thanks to these two functions, we can now define formally signal transitions.

\begin{definition}[Signal transition]\label{tcom}
  Let $\ran \in \tau \cup \{\bot\}$ be the task that has just been executed if any, $\bot$ otherwise.
  Let $\theta \in \{\true,\false\}$ denote whether $\ran$ signals completion explicitly.
  Then, we say that $S^\scom = \systemstate[\scom]$ is a \emphindef{signal successor} of
  $S$ under $\ran$, denoted $S \CpTrans{\ran,\theta} S^\scom$ iff:
  \begin{numcases}
  {S^\scom = }
        S & $\begin{multlined}[c]
            \mathrm{if}\ \ran \in \Completed{S} \cup \lbrace \bot \rbrace \\
            \vee \left(\neg\theta \wedge \rct_S(\ran) > 0\right) \\
            \vee \left(\theta \wedge \rct_S(\ran) = 0\right)
        \end{multlined}$ \\
        \sigCmp{S}{\ran} & $\begin{multlined}[c]
            \mathrm{if}\ \ran \notin \Completed{S} \cup \lbrace \bot \rbrace \\
            \wedge \theta \wedge \rct_S(\ran) > 0\text{.}
        \end{multlined}$ \\
        \critUp{S}{\ran} & $\begin{multlined}[c]
            \mathrm{if}\ \ran \notin \Completed{S} \cup \lbrace \bot \rbrace \\
            \wedge \neg\theta \wedge \rct_S(\ran) = 0\text{.}
        \end{multlined}$ 
  \end{numcases}
\end{definition}

The cases' number match the outcomes previously described.
Note that $S$ always has itself has a signal successor.
Then, when $\ran \not\in \Completed{S} \cup \lbrace \bot \rbrace$, $S$ has another signal successor which is either
$\sigCmp{S}{\ran}$ or $\critUp{S}{\ran}$ depending on the setup.
%
%

Finally, we define the automaton $A(\tau, \schedule)$ that formalises the behaviour of
the system of dual-criticality sporadic task set $\tau$, when executed under a scheduling algorithm $\schedule$:

\begin{definition}\label{auto}
  Given a system of dual-criticality sporadic tasks $\tau$ and a scheduler $\schedule$, the
  automaton $A(\tau,\schedule)$ is the tuple $\langle V, E, v_0, F \rangle$ where:

  \begin{enumerate}

    \item $V=\States{\tau}$
    \item $(S_1,S_2) \in E$, iff there are $S^\sreq, S^\srun \in \States{\tau}$,
      $\tau^\sreq \subseteq \tau$ and $\theta \in \lbrace\true,\false\rbrace$ \sth{}:
        $S_1 \RqTrans{\tau^\sreq} S^\sreq
        \CtTrans{\schedule(S^\sreq)} S^\srun \CpTrans{\schedule(S^\sreq),\theta} S_2$
    \item $v_0 = \langle \rct_{v_0}, \nat_{v_0}, \crit_{v_0}\rangle$ where $\crit_{v_0} =
      \lo$ and $\forall \tau_i \in \tau: \nat_{v_0}(\tau_i) = \rct_{v_0}(\tau_i) = 0$
    \item $F = \DeadlineMiss{\tau}$

  \end{enumerate}

\end{definition}

As we assumed that $\schedule$ must be deterministic, $\run$ and $\ran$, equalling both to $\schedule(S^\sreq)$, will be the same, as intended by the definition of the automaton.
Each possible execution of the task set corresponds to a path in $A(\tau, \schedule)$ and vice
versa.
States in $\DeadlineMiss{\tau}$ correspond to states of the system where a deadline is or
has been missed. Hence, the set of dual-criticality sporadic tasks $\tau$ is feasible under
scheduler $\schedule$ iff $A(\tau,\schedule)$ is \textit{safe}, \ie $\DeadlineMiss{\tau}$ is not
reachable in $A(\tau,\schedule)$.
\begin{figure}
  \centering
    \includegraphics[width=0.9\columnwidth]{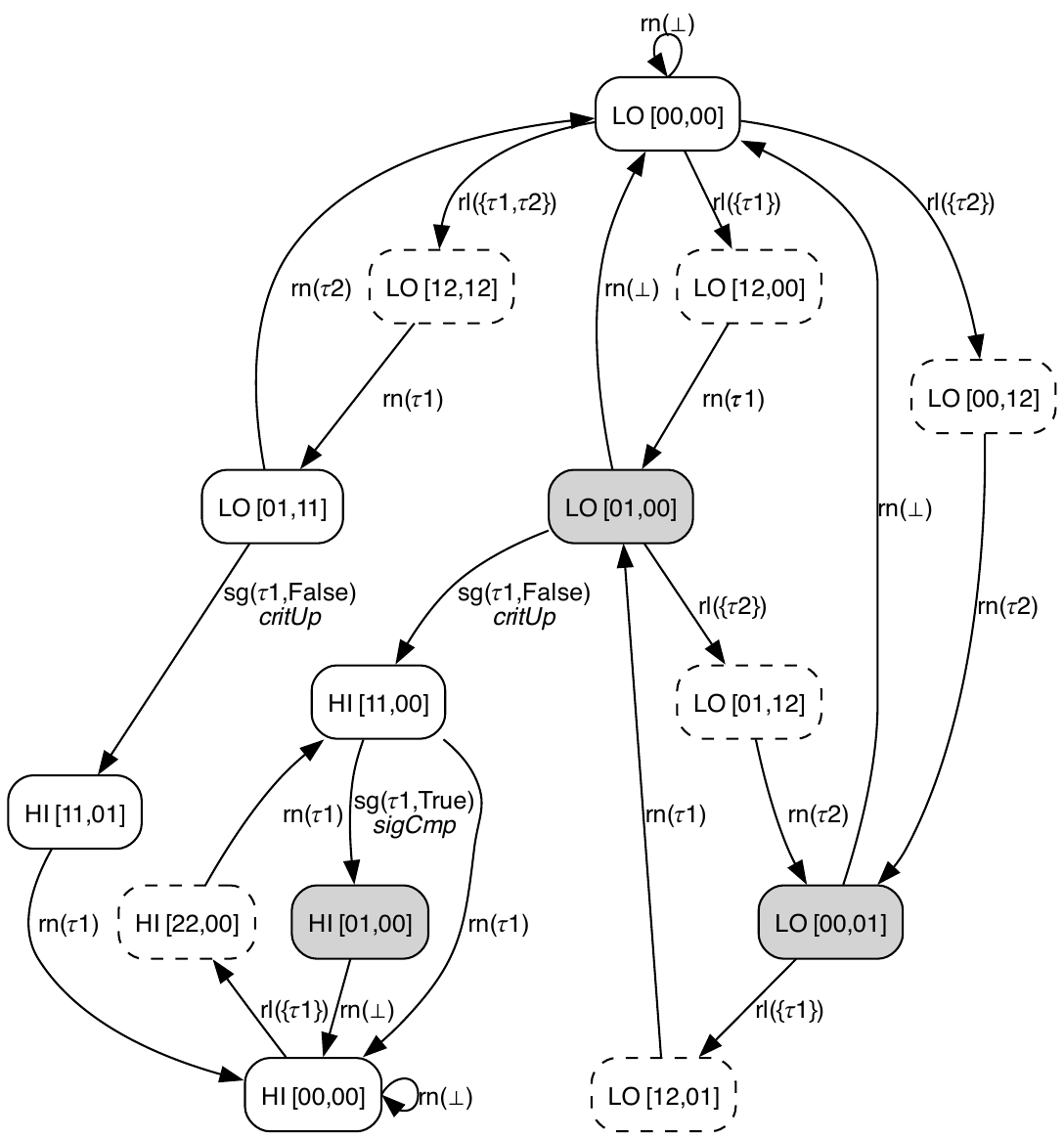}
    \caption{$A(\tau^{a}, \edfvd)$ developed automaton with intermediary transitions. Intermediary states have a dashed outline, and greyed out states are simulated.}
    \label{fig:automaton}
\end{figure}
\autoref{fig:automaton} illustrates such an automaton,
representing the possible execution of a task set scheduled with the EDF-VD scheduler.
In this example, the automaton depicts the dual-criticality sporadic task set $\tau^{a}$ as specified in~\autoref{sec:problemdef}.
System states are represented by nodes.
For the purpose of saving space, we represent a state $S$ with the $\chi[\alpha\beta,
\gamma\delta]$ format, meaning $\crit_S=\chi$, $\rct_S(\tau_1) = \alpha$, $\nat_S(\tau_1) =
\beta$, $\rct_S(\tau_2) = \gamma$ and $\nat_S(\tau_2) = \delta$.
We explicitly represent run transitions by edges labelled with $\mathsf{rn}$, signal
transitions by edges labelled with $\mathsf{sg}$, and release transitions by edges labelled with
$\mathsf{rl}$.
The $\tau^\sreq=\emptyset$ release loops and $S^\scom=S$ signal
loops are omitted for readability.
Notice that the graph in \autoref{fig:automaton} is a developed way of
representing $A(\tau^{a}, \edfvd)$ as it is built with intermediary
transitions.  On the other hand, \autoref{fig-full-automaton} shows a
fully developed automaton without intermediary transition for the task
set $\tau^a$ defined in \autoref{sec:problemdef}.

\begin{figure}
    \centering
    \includegraphics[width=0.4\columnwidth]{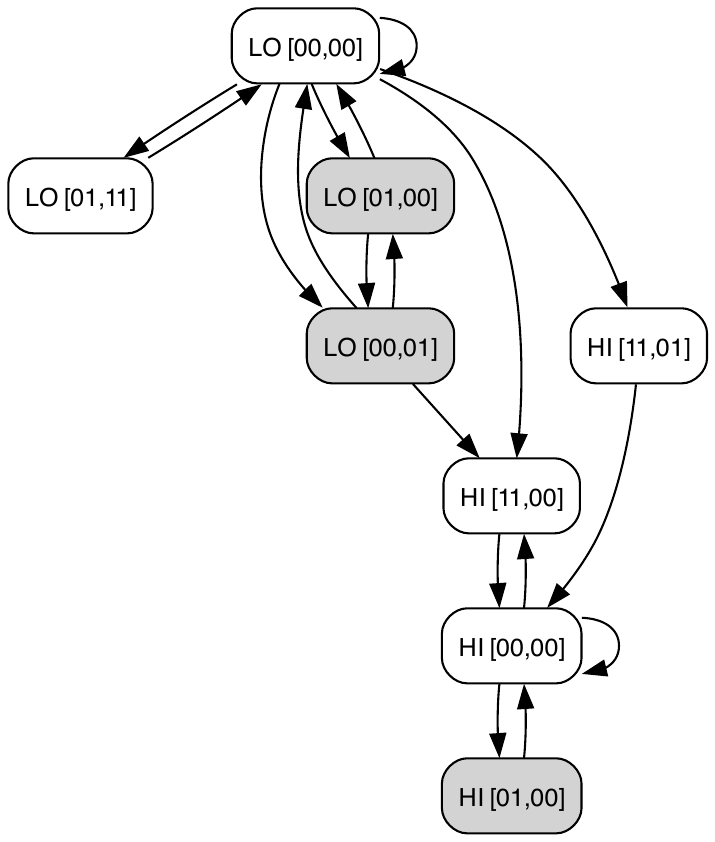}
    \caption{Example of a fully developed automaton, without intermediary transitions. Greyed out states are simulated.}
    \label{fig-full-automaton}
\end{figure}

We can observe different behaviour from the system.
There might be several release transitions from a single state, as it can be seen from the initial node $\lo [00,00]$.
Indeed, four combinations of releases are possible.
When $\tau^\sreq = \emptyset$, there is a release transition from the state to itself, which is omitted on the figure.
The signal transition from the state $\lo [01,11]$ contains a mode change which leads to a $\hi$-criticality system state and the $\lo$-criticality task $\tau_2$ is discarded.
Finally, in the signal transition from the state $\hi [11,00]$, the task $\tau_1$ signals completion explicitly before exhausting all its execution time budget $C_i(\hi) = 2$, and its $\rct$ is set to $0$.
Per the automaton definition, a signal transition exists only when it is preceded by a run transition, emanating from $\hi [22,00]$ in this case.
Thus, the first time that the state $\hi [11,00]$ was visited via a signal transition from $\lo [01,00]$, the signal transition to $\hi [01,00]$ was not yet generated.

\section{An efficient algorithm for safety with simulation relation and oracles}\label{sec:an-effic-algor}

In this section, we present an efficient algorithm to check for safety in an automaton,
which is based on two hypothesis:
\textit{(i)} we can query an \emph{oracle} which can sometimes tell us whether a given state can
or cannot reach the set of failure states $F$;
\textit{(ii)} we have at our disposal a \emph{simulation relation} between states. Intuitively,
when $v'$ simulates $v$, it means that all paths that can occur from $v$
can be witnessed by a path from $v'$. In particular, when $v$ can reach the
failure states, then $v'$ can reach them too.

Both the oracle and the simulation relation can improve the performance of an algorithm answering the
safety problem.
Whenever the oracle guarantees that $v$ cannot reach the failure states, there is no need
to compute the paths starting in $v$ (and it is safe to do so, as no path to the failure
states will be dropped).
When the oracle says that a state $v$ (that is reachable from $v_0$) can reach the failure
states $F$, we can immediately conclude that $F$ is reachable from $v_0$ as well, hence
that the automaton is \emph{unsafe}.
Finally, when two reachable states $v$ and $v'$ are found \sth{} $v'$ simulates $v$, only
the paths starting from $v'$ need to be computed.

We show in \autoref{sec:simul-relat-oracl} how such an oracle and such a simulation
relation can be derived in our mixed criticality setting, but the generality of the
concepts we explain here make them applicable to other scheduling problems.
We formalise the notion of oracle:

\begin{definition}[Safe and unsafe states]\label{def:safe-unsafe-states}
  Let $A = \langle V, E, v_0, F \rangle$ be a finite automaton.
  Then, a state $v$ is \emphindef{safe} iff there is no path starting in $v$ that reaches $F$,
  \ie{} $\Reach{v} \cap F = \emptyset$.

  A state $v$ that is not safe (\ie{} there is a path from $v$ that reaches $F$)
  is called unsafe.

\end{definition}

In order to formalise the oracle, we assume that we have at our disposal a set $\Safe$
containing only safe states and a set $\Unsafe$ containing only unsafe states.
Without loss of generality, we assume that $F\subseteq\Unsafe$.
Observe that we can have $\Safe=\emptyset$ and $\Unsafe=F$, so we do not request that
$\Safe\cup\Unsafe=V$.
However, the bigger those sets, the more efficient our algorithm will potentially be.
Next, we define the notion of \emph{simulation relation}.

\begin{definition}[Simulation relation]\label{def:simulation}
  Let $A = \langle V, E, v_0, F \rangle$ be a finite automaton.
  A relation $\simrel\subseteq V\times V$ is a \emphindef{simulation relation} iff it
  respects the following properties:

  \begin{enumerate}

    \item $\simrel$ is a preorder, \ie it is reflexive and transitive.
    \item for all states $v_1$ and $v_1'$ \sth{} $v_1\simrel v_1'$, for all $v_2$
      \sth{} $(v_1,v_2)\in E$, there exists $v_2'\in V$ \sth{}:
        \textit{(i)} $(v_1',v_2')\in E$;
        \textit{(ii)}  and $v_2\simrel v_2'$.

    \item for all states $v$ and $v'$ \sth{} $v\simrel v'$: $v\in F$ implies that
      $v'\in F$.

  \end{enumerate}

\end{definition}

Whenever $v_1\simrel v_1'$, we say that $v_1'$ \emph{simulates} $v_1$.
Thus, the definition says that, for every move $(v_1,v_2)$ that can be performed from
$v_1$, there is a matching move $(v_1',v_2')$ from $v_1'$.
Here, `matching' means that the state $v_2'$ simulates $v_2$
($v_2\simrel v_2'$) and that $v_2'$ is a failure state if $v_2$ is.
From this definition, we establish the following:

\begin{proposition}\label{prop:simulate-safe}
  Let $A = \langle V, E, v_0, F \rangle$ be a finite automaton.
  For all pairs of states $v$ and $v'$ \sth{} $v\simrel v'$, the following holds:
  (1) if $v$ is unsafe, then $v'$ is unsafe too;
  (2) if $v'$ is safe, then $v$ is safe too.
\end{proposition}
\begin{proof}
  First, let us assume $v$ is unsafe.
  From \autoref{def:safe-unsafe-states}, it means that there exists a path
  $v,v_1,v_2,\ldots, v_k$ \sth{} $v_k\in F$.
  By applying inductively \autoref{def:simulation}, we can build a path
  $v',v_1',v_2'\ldots,v_k'$ \sth{} $v_i\simrel v_i'$ for all $1\leq i\leq k$.
  By \autoref{def:simulation} again, and since $v_k\in F$, we conclude that $v_k'\in
  F$ too.
  Hence, $v'$ is unsafe.
  This proves point~1.
  Point~2 is simply its contrapositive.
\end{proof}

This observation prompts for the definition of the set of `all states that are simulated
by/simulate a given state $v$'.
Given a state $v$ and a simulation relation $\simrel$, we let the
\emph{upward-closure} of $v$ be the set of all states that simulate $v$, \ie
$\uc{v}=\{v'\mid v\simrel v'\}$; and its \emph{downward-closure} be $\dc{v}=\{v'\mid
v'\simrel v\}$ (the set of all states that are simulated by $v$).
The upward and downward-closures are also defined on a set of states, for which it becomes the union of the respective closures of all the states within the set.
Finally, given a set of states $Y\subseteq V$, we let $\Max{Y}=\{v\in S\mid \nexists v'\in
S: v\simrel v'\}$ be the set of \emph{maximal} states of $S$ w.r.t. the simulation
$\simrel$, \ie those states in $S$ that are not simulated by any other state in $S$.

Let us now explain how these notions (sets of safe and unsafe states, and simulation
relation) can be exploited in a generic algorithm to solve the safety problem in a finite
automaton.
First, let us recall the general strategy of \emph{breadth first search} to solve safety in
an automaton.
It can be formalised as computing two sequences of sets $\left(R_i\right)_{i\geq 0}$ and
$\left(N_i\right)_{i\geq 0}$ \sth{}, for all $i\geq 0$, $R_i$ and $N_i$ are the sets of all
states that can be reached from $v_0$ in at most $i$ steps and exactly $i$ steps,
respectively.
We denote by
$\Succ{Y} = \lbrace v' | \exists (v, v') \in E  \, \forall v \in Y \rbrace$
the set of all successors of all states in $Y \subseteq V$,
and we have
  $N_0 = \{v_0\}$, $\forall i\geq 0: N_{i+1} = \Succ{N_i}\setminus R_i$,
  $R_0 = \{v_0\}$, and $\forall i\geq 0: R_{i+1} = R_i\cup N_{i+1}$.
Then, an algorithm to decide reachability consists in computing $N_1,R_1,N_2,R_2,\ldots
N_j,R_j,\ldots$ until:
\textit{(i}) either $N_j\cap F\neq\emptyset$, in which case $F$ is reachable from $v_0$; or
\textit{(ii}) $R_j=R_{j+1}$, in which case $R_j$ contains all the reachable states and we know
    that $F$ is not reachable (otherwise, we would have already returned `fail').
In order to improve this algorithm, we compute the sequences
$\left(\Rtilde_i\right)_{i\geq 0}$ and $\left(\Ntilde_i\right)_{i\geq 0}$:
%
%
\begin{align*}
  \Ntilde_0 = \Rtilde_0 = & \{v_0\}\setminus\dc{\Safe} \\
  \forall i\geq 0: \Ntilde_{i+1} = & \Max{\Succ{\Ntilde_i}\setminus
  \dc{\Rtilde_i\cup \Safe}} \\
  \forall i\geq 0: \Rtilde_{i+1} = & \Max{\Rtilde_i\cup \Ntilde_{i+1}}.
\end{align*}

Intuitively, these sequences contain less elements than the original ones, but still retain enough information to solve the safety problem. Let us explain why. Assume the algorithm has computed set $\Rtilde_i$ and $\Ntilde_i$, 
and let us consider how it computes $\Ntilde_{i+1}$. First, the algorithm computes the successors
of $\Ntilde_i$, then computes 
$\Succ{\Ntilde_i}\setminus \dc{\Rtilde_i\cup \Safe}$, i.e. it removes, from the successors 
of $\Ntilde_i$, all the elements that are simulated by an element from $\Rtilde_i$ or from $\Safe$. 
Finally, it keeps only the maximal elements from this resulting set.

To understand why these optimisations are correct, assume there is a state $v\in \Succ{\Ntilde_i}$ which
is also simulated by an element of $\Safe$ (hence $v\in\dc{\Safe}$).
With our optimisations, $v$ is not in $\Ntilde_{i+1}$, and its successors will 
not be computed at the next step.
However, this is not a problem, because $v$ is simulated by a safe state, hence it is
\emph{safe too} by \autoref{prop:simulate-safe}.
Similarly, if $v\in\dc{\Rtilde_i}$, it means we have already computed all the necessary successors at a previous step, keeping only those that are potentially unsafe.
So, it is correct to avoid computing the successors of $v$, since we are looking for paths
that lead to \emph{unsafe} states.
Now, assume that there is $v\in\Succ{\Ntilde_i}\setminus \dc{\Rtilde_i\cup \Safe}$ which is not $\simrel$-maximal.
This implies that there is $v'\in \Ntilde_{i+1}$ that simulates $v$.
There are two possibilities. Either $v$ is \emph{safe}, and we do not need to compute its
successors.
Or $v$ is \emph{unsafe}, but then, so is $v'$.
So, again, we keep in $\Ntilde_{i+1}$ enough information to discover a path to the failure
states. With these definitions, all sets $\Rtilde_i$ and $\Ntilde_i$ are \emph{$\simrel$-antichains of
$\simrel$-maximal states} (\ie sets of states which are all incomparable and maximal
w.r.t. $\simrel$) that retain enough information to find paths to $F$.

Based on these sequences, we propose \autoref{acbfs-safe} to check for safety in a
given finite automaton $A = \langle V, E, v_0, F \rangle$, when we have at our disposal a simulation
relation $\simrel$ on $V$, and two sets $\Safe$ and $\Unsafe\supseteq F$ of safe and
unsafe states respectively, provided to us by an Oracle.
This algorithm consists in computing $\Ntilde_0, \Rtilde_0, \Ntilde_1, \Rtilde_1,\ldots$
until either
\textit{(i}) $\Ntilde_i\cap \uc{\Unsafe}\neq \emptyset$, in which case we know that $F$ is
reachable; or
\textit{(ii}) $\Ntilde_i=\emptyset$, in which case we have explored enough states to guarantee
that $F$ is not reachable.


\begin{algorithm}
  \begin{algorithmic}[1] 
    \State $i \leftarrow 0$
    \State $\Rtilde_0 \leftarrow \{v_0\}\setminus\dc{\Safe}$
    \State $\Ntilde_0 \leftarrow \{v_0\}\setminus\dc{\Safe}$
    \Repeat
      \If{$\Ntilde_i\cap\uc{\Unsafe}\neq\emptyset$}
        \State \Return \texttt{Fail}
      \EndIf
      \State $\Ntilde_{i+1}\leftarrow \Max{\Succ{\Ntilde_i}\setminus
          \dc{\Rtilde_i\cup \Safe}}$
      \State $\Rtilde_{i+1}\leftarrow \Max{\Rtilde_i\cup \Ntilde_{i+1}}$ 
      \State $i \leftarrow i+1$
    \Until{$\Ntilde_i=\emptyset$}
    \State \Return \texttt{Safe}
  \end{algorithmic}
  \caption{Antichain breadth first search with safe and unsafe states optimisation}\label{acbfs-safe}
\end{algorithm}


\begin{proposition}\label{prop:algo-correct-and-terminates}
  On all automata $A$, \autoref{acbfs-safe} terminates and returns `\texttt{Fail}'
  iff $A$ is unsafe.
\end{proposition}
For the sake of readability, we present here a proof sketch. The
complete proof is in Appendix~\ref{app:proof-algo}

\begin{proof}[Proof Sketch]
  We first establish termination by showing that all states computed in some $\Rtilde_i$ are
  reachable states, \ie $\Rtilde_i\subseteq \Reach{A}$ for all $i\geq 0$.
  We then observe that the sets $\dc{\Rtilde_i}$ keep growing, \ie
  $\dc{\Rtilde_0}\subseteq\dc{\Rtilde_1}\subseteq\cdots\subseteq\dc{\Rtilde_i}\subseteq\cdots$.
  However, since $\dc{\Rtilde_i}\subseteq \dc{\Reach{A}}$ for all $i\geq 0$, and since
  $\dc{\Reach{A}}$ is a finite set, this sequence must eventually stabilise, \ie
  $\dc{\Rtilde_\ell}=\dc{\Rtilde_{\ell+1}}$ for some $\ell$.
  From this, we can conclude that $\Ntilde_{\ell+1}=\emptyset$ and the algorithm
  terminates.

  Next, we prove soundness, \ie when the algorithm returns `\texttt{Fail}', we have indeed
  $\Reach{A}\cap F\neq \emptyset$.
  This stems again from the fact that $\Ntilde_i\subseteq \Reach{A}$ for all $i\geq 0$.
  Hence, when we return `\texttt{Fail}', we have indeed found a reachable state that is
  unsafe.

  Finally, we establish completeness, \ie when $\Reach{A}\cap F\neq \emptyset$, the
  algorithm returns `\texttt{Fail}'.
  To obtain this result, we consider one path $v_0,v_1,\ldots, v_k, v_{k+1}$ \sth{}, $v_{k+1}\in\uc{\Unsafe}$, and $v_i\not\in\uc{\Unsafe}$ for all $0\leq i\leq k$.
  That is, the path reaches $\uc{\Unsafe}$ in its last state only.
  We show that the prefix $v_0,v_1,\ldots, v_k$ can be `found' in the sequence
  $\Rtilde_0,\ldots, \Rtilde_k$ in the following sense: $0\leq i\leq k$:
  $v_i\in\dc{\Rtilde_i}$.
  From that, we conclude that $\Ntilde_{k+1}\cap \uc{\Unsafe}\neq\emptyset$, and that
  $\Ntilde_{k+1}$ will be computed by the algorithm, hence returning `\texttt{Fail}'.
\end{proof}


\section{A simulation relation and oracles for dual-criticality scheduling}
\label{sec:simul-relat-oracl}

In order to leverage the optimised algorithm of \autoref{sec:an-effic-algor} on the
automaton defined in \autoref{sec:automaton-semantic} (to model the behaviour of dual-criticality
task set),
we define in this Section a suitable simulation relation and \textit{ad hoc} oracles.
In \autoref{sec:evaluation}, we will show that these optimisations are efficient in practice to reduce
the size of the state space that needs to be explored.

\subsection{Idle tasks simulation relation}
We first define a simulation relation $\idle$, called the \textit{idle tasks
simulation relation} that can be computed efficiently by inspecting the values 
$\nat$, $\rct$ and $\crit$ stored in the states.

\begin{definition}\label{def:idle}
  Let $\tau$ be a set of dual-criticality sporadic tasks.
  Then, the \emphindef{idle tasks preorder} $\idle \subseteq \States{\tau} \times
  \States{\tau}$ is \sth{} for all $S_1, S_2: S_1 \idle S_2$ iff:
    \textit{(i)} $\crit_{S_2} = \crit_{S_1}$; \textit{(ii)} $\rct_{S_2} = \rct_{S_1}$;
    \textit{(iii)} for all $\tau_i$ \sth{} $\rct_{S_1}(\tau_i) = 0 : \nat_{S_2}(\tau_i) \leq
      \nat_{S_1}(\tau_i)$; and \textit{(iv)} for all $\tau_i$ \sth{} $\rct_{S_1}(\tau_i) > 0 : \nat_{S_2}(\tau_i) =
      \nat_{S_1}(\tau_i)$.
\end{definition}

This relation is transitive and reflexive, so it is indeed a preorder.
The relation also defines a partial order on $\Active{S}$, because it is antisymmetric.
Note that $S_1 \idle S_2$ implies that $\Active{S_2} = \Active{S_1}$ since $\rct_{S_2} = \rct_{S_1}$.
%
%
We show that this preorder is indeed a simulation relation for a
deterministic scheduler:

\begin{theorem}\label{theo:idle-sim}
  Let $\tau$ be a dual-criticality sporadic task system and $\schedule$ a deterministic
  uniprocessor scheduler for $\tau$.
  Then $\idle$ is a simulation relation for $A(\tau, \schedule)$.

\end{theorem}
For the sake of readability we present here a proof sketch. The
complete proof can be found in Appendix~\ref{app:proof-simu}.

\begin{proof}[Proof Sketch]
  Let $S_1, S^\scom_1$ and $S_2$ be three states in $\States{\tau}$ s.t. $(S_1,
  S^\scom_1) \in E$ and $S_1 \idle S_2$, let us show that there exists $S^\scom_2 \in
  \States{\tau}$ with $(S_2, S^\scom_2) \in E$ and $S^\scom_1 \idle S^\scom_2$.

  Let $S^\sreq_1$ and $S^\srun_1 \in \States{\tau}$, $\tau^\sreq \subseteq \tau $ and $\theta \in \lbrace\true,\false\rbrace$ be such that : $S_1 \RqTrans{\tau^\sreq} S^\sreq_1
  \CtTrans{\schedule(S^\sreq_1)} S_1^\srun \CpTrans{\schedule(S^\sreq_1),\theta} S_1^\scom$. Those exist by
  \autoref{auto} and since $(S_1, S_1^\scom) \in E$.

  First, observe that, by \autoref{def:idle}, $\Eligible{S_1}\subseteq\Eligible{S_2}$. 
  So, by \autoref{treq}, there exists a $\tau^\sreq$-release transition 
  from $S_2$. Let $S_2^\sreq$ be the state s.t. $S_2 \RqTrans{\tau^\sreq} S_2^\sreq$. 
  By \autoref{treq}, again, it is easy to check that $S_1^\sreq\idle S_2^\sreq$.  
  
    Next, let us show that we can simulate the run transition from $S_2^\sreq$. To alleviate notations, 
    let us denote, from now on, 
    $\schedule(S^\sreq_1)$ by $\run_1$ and $\schedule(S^\sreq_2)$ by $\run_2$. Since the scheduler is deterministic, we have $\run_1=\run_2$,
    by \autoref{run}. Hence, we can let $S_2^\srun$ be the state s.t. $S^\sreq_2
    \CtTrans{\run_2} S_2^\srun$. Finally, by \autoref{def:idle} and~\ref{run}, 
    it is easy to check that $S_1^\srun\idle S_2^\srun$.

 Then, let us show that we can simulate the signal transitions from $S_1^\srun$. Let $S_2^\scom$ be the state s.t. $S^\srun_2
    \CpTrans{\run_2,\theta} S_2^\scom$ and let us observe $S_1^\scom$. By \autoref{tcom}, there are three possible values for $S_1^\scom$ which are all simulated by $S_2^\scom$.
    Indeed, keeping in mind that, $\run_1=\run_2$, $S^\srun_1 \idle S^\srun_2$ is already established.
    Then, it is easy to verify that $\sigCmp{S^\srun_1}{\run_1}\idle\sigCmp{S^\srun_2}{\run_2}$ and $\critUp{S^\srun_1}{\run_1}\idle \critUp{S^\srun_2}{\run_2}$.

  It remains to prove that \emph{if} $S_1 \idle S_2$ and $S_1 \in
  \DeadlineMiss{\tau}$, \emph{then} $S_2 \in \DeadlineMiss{\tau}$ too.
  Let $\tau_i$ be such that $\rct_{S_1}(\tau_i) > 0 \wedge \ttd_{S_1}(\tau_i) \leq 0$.
  Since $S_1 \idle S_2$ : $\rct_{S_2}(\tau_i) = \rct_{S_1}(\tau_i)$ and $\nat_{S_2}(\tau_i)
  \leq \nat_{S_1}(\tau_i)$, thus $\ttd_{S_2}(\tau_i) \leq \ttd_{S_1}(\tau_i)$.
  Hence $\rct_{S_2}(\tau_i) > 0 \wedge \ttd_{S_2}(\tau_i) \leq \ttd_{S_1}(\tau_i)
  \leq 0$ and therefore $S_2 \in \DeadlineMiss{\tau}$ as per~\autoref{fail}.
\end{proof}

\autoref{fig:automaton}, presented in \autoref{sec:automaton-semantic},
illustrates the effect of $\idle$ with \autoref{acbfs-safe}.
If a state $S_2$ has been encountered previously, and we find another state $S_1$ \sth{} 
$S_1 \idle S_2$, then we can avoid exploring $S_1$ and its successors.
However, this does not mean never encountering a successor of $S_1$ as they
may be encountered through other paths, or have been encountered already.
In \autoref{fig:automaton}, grey states can be avoided as they are simulated by another
state: $\hi [01,00] \idle \hi [00,00]$ and, $\lo [00,01]$ and $\lo [01,00]$ are simulated by $\lo [00,00]$.

\subsection{Safe state oracles}

We present a safe state oracle, which is a sufficient schedulability condition depending on the state of the system.


\begin{oracle}[$\hi$ idle point]
$\{S \mid \crit_{S} = \hi \wedge  \Active{S} = \emptyset \} \subseteq \Safe$.
\begin{proof}
As per the \emph{``due diligence''} outlined in \autoref{sec:problemdef}, the task set comprising only $\hi$ tasks, where
$\forall \tau_i: C_i = C_i(\hi)$, must be schedulable. Hence, as of reaching an idle point in $\hi$ mode, no more deadline misses are possible.
\end{proof}
\end{oracle}

\subsection{Unsafe state oracles}

%
In this section, we present several \emph{necessary} conditions depending on the states of the system
and provide their equivalent formulation in terms of unsafe oracles. Note that the oracles presented below are scheduler agnostic.
%

\begin{definition}[Laxity]
    The \emphindef{laxity} of an active task $\tau_i$ in the state $S$ is: $\laxity_S(\tau_i) = \ttd_S(\tau_i) - \rct_S(\tau_i)$.
\end{definition}

\begin{lemma}\label{lemma:nc-poslax}
    For feasibility it is \emph{necessary} to have
    $\forall \tau_i \in \Active{\tau}: \laxity_S(\tau_i) \geq 0$.
\end{lemma}

\begin{proof}
    It is obvious that if $\ttd_S(\tau_i) < \rct_S(\tau_i)$, then even if the processor is given to the job immediately until its deadline, we will miss its deadline.
\end{proof}

\begin{oracle}[Negative laxity]
$\{S \mid \exists \tau_i \in \Active{\tau}: \laxity_S(\tau_i) < 0  \} \subseteq \Unsafe$.
\begin{proof}
It holds trivially from \autoref{lemma:nc-poslax} that any state $S$ violating the necessary condition will lead to a deadline miss.
\end{proof}
\end{oracle}
In $\lo$ mode, we introduce a \emph{stronger} condition that anticipates an imminent mode change. 
The \emph{worst-laxity} of a $\lo$ task is simply its laxity.
    However, for a $\hi$ task, we incorporate its ``bonus'' execution time ($C_i(\hi)-C_i(\lo)$) that would be incurred if it does not explicitly signal completion, eventually causing a mode change.

\begin{definition}[Worst-laxity]
    The \emphindef{worst-laxity} of an active task $\tau_i$ in state $S$ is: $\wlaxity_S(\tau_i) \defeq \laxity_S(\tau_i) - (C_i(L_i)-C_i(\crit_S))$.
\end{definition}


\begin{lemma}\label{lemma:worst-laxity}
    $\forall \tau_i$, it is \emph{necessary} to have 
        $\wlaxity_S(\tau_i) \geq 0 \text{.}$
\end{lemma}

\begin{proof}
    In $\hi$ mode the laxity is equivalent to the notion of worst-laxity, the property follows from \autoref{lemma:nc-poslax}. In $\lo$ mode we have to distinguish between $\lo$ and $\hi$ tasks. 
    For a $\lo$ task it is necessary to have non negative laxity. For a $\hi$ task if this condition is not satisfied and if a mode change occurs \emph{simultaneously}, by definition the remaining computation time of the active job of $\tau_i$ is larger than the time we have before reaching the deadline. Therefore, we will inevitably miss the task deadline, which proves the property.
\end{proof}

\begin{oracle}[Negative worst-laxity]
    $\{S \mid \exists \tau_i \in \Active{\tau}: \wlaxity_S(\tau_i) < 0\} \subseteq \Unsafe$.
    \begin{proof}
        It holds trivially from \autoref{lemma:worst-laxity} that any state $S$ violating the necessary condition will lead to a deadline miss.
    \end{proof}
\end{oracle}
For the next oracle, we propose to adapt the reasoning of demand bound functions (\dbf)
--- as leveraged for EDF and its variations in many prior works~\cite{baruah1990} ---
to formulate a demand function related to the state of the system
--- \ie the \crit{}, \rct{} and \nat{} information.
Instead of defining, per-task $\tau_i$, a demand function between two \emph{absolute} instants $t_1$ and $t_2$,
we define a demand function \emph{relative} to the current state,
with a single parameter $t$
which is a time in the future \emph{relative} to the ``current time'' (or current instant) in $S$.

We need to separate concerns between current jobs
--- that can be disabled (\lo{} case) or extended (\hi{} case) ---
and future jobs.
We define the number of future jobs of a task,
\ie{} the number of jobs the task can release until this future instant.

\begin{definition}[Number of future jobs of a task until a future instant]
    the \emphindef{maximum number of jobs} that $\tau_i$ can release \emph{strictly} after the current instant in $S$ and before a future instant $t$,
    and that have a deadline no later than $t$, assuming $S$ is in mode $\alpha$,
    is:
    \begin{numcases}{\nj_{S}^{\alpha}(\tau_i,t) \defeq}
        0                       &  \textnormal{if} $L_i<\alpha$ \nonumber \\
        \biggl \lfloor \frac{%
            \max
            \lbrace
                t - \ttd_S(\tau_i)
                ,
                0
            \rbrace
        }{%
            T_i
        } \biggr \rfloor        &  \textnormal{otherwise.} \nonumber
    \end{numcases}
\end{definition}

The first case represents the fact that no job can be released by a \lo{} task when assuming \hi{} mode.
The second case is the amount of periods that fits within the interval from the deadline of the current job and $t$.
%
Based on that, we can define the demand function.

\begin{definition}[Demand function of a task in a state]
    The \emphindef{demand function} of a task $\tau_i$ in state $S$ in mode $\alpha$ is a lower bound on the maximum amount of computation time required by jobs of $\tau_i$
    between the current instant in $S$ and a future instant $t$,
    only for jobs of $\tau_i$ having deadlines before $t$,
    assuming $S$ is in mode $\alpha$
    or that the mode $\alpha$ will be reached at the next possible instant.
    Formally, with $\alpha \in \{\lo{}, \hi{}\}$: $\df_{S}^{\alpha}(\tau_i, t)$
    \begin{numcases}{\defeq}
        0                       &  \textnormal{if} $t<\ttd_S(\tau_i) \vee L_i<\alpha$ \nonumber \\
        \nj_{S}^{\alpha}(\tau_i,t) \cdot C_i(\alpha) &  \textnormal{else if} $\rct_S(\tau_i) = 0$ \nonumber \\
        \begin{multlined}[c]
            \nj_{S}^{\alpha}(\tau_i,t) \cdot C_i(\alpha) + C_i(\alpha) \\ - C_i(\crit_S) + \rct_S(\tau_i)
        \end{multlined}
        &  \textnormal{otherwise.} \nonumber
    \end{numcases}
\end{definition}

No work must be accounted in the demand $\df_{S}^{\alpha}(\tau_i, t)$
if $t$ is before the deadline of $\tau_i$ or if the criticality of $\tau_i$ is below the considered mode $\alpha$.
Otherwise, if $\tau_i$ is idle, we only consider the demand of \emph{future} jobs,
and do not account for the ``bonus'' execution time ($C_i(\hi{}) - C_i(\lo{})$)
a running job of a $\hi{}$ task receives if a mode switch occurs.
This is accounted in the third case, together with the remaining time ($\rct_S(\tau_i)$) of the current job.
%
%
%
Notice that $\nj_{S}^{\crit_S}(\tau_i,\ttd_S(\tau_i)) = 0$, so $\df_{S}^{\crit_S}(\tau_i, \ttd_S(\tau_i)) = \rct_S(\tau_i)$.
We can now define the global demand bound.
\begin{definition}[Demand bound function in a state]
    The global \emphindef{demand bound function} is the total maximum amount of computation time required
    between the current instant in $S$ and a future instant $t$
    by jobs of the task set $\tau$,
    only for jobs having deadlines before $t$, assuming mode $\alpha$.
    Formally: $\dbf_{S}^{\alpha}(t) = \sum_{\tau_i \in \tau} \df_{S}^{\alpha}(\tau_i,t)$.
\end{definition}

The definitions of $\nj$, $\df$ and $\dbf$ uses the mode parameter $\alpha$.
This allows us to consider the current mode when we set $\alpha = \crit_S$
and to anticipate the demand if a mode switch occurs when we set $\alpha = \hi{}$.

\begin{lemma}\label{lemma:nc-dbf}
    For feasibility it is \emph{necessary} to have
    $
    \forall \tau_i \in \Active{\tau}:
    \ttd_S(\tau_i) \geq \dbf_{S}^{\crit_S}(\ttd_S(\tau_i))
    $.
\end{lemma}

\begin{oracle}[Over demand]\label{oracle:over-demand}
    $\{S \mid \exists \tau_i \in \Active{\tau}: \ttd_S(\tau_i) < \dbf_{S}^{\crit_S}(\ttd_S(\tau_i))\} \subseteq \Unsafe$.
    \begin{proof}
        For any $\tau_k \in \tau$,
        $\nj_{S}^{\crit_S}(\tau_k,\ttd_S(\tau_i))$ is the highest possible $\ell$ such that
        $\nat_S(\tau_k) + D_k + (\ell-1) \cdot T_k \leq \ttd_S(\tau_i)$,
        showing it is the highest possible number of jobs released by $\tau_k$ between the current instant in $S$ and the next deadline of $\tau_{i}$
        having a deadline before $\ttd_S(\tau_i)$.
        Note that $\df_{S}^{\crit_S}(\tau_k, \ttd_S(\tau_i))$ represents the total amount of execution time required by $\tau_k$ including its potentially active job and all its future jobs until $\ttd_S(\tau_i)$, assuming no mode change.
        Consequently, $\dbf_{S}^{\crit_S}(\ttd_S(\tau_i))$ is the amount of computation time required by all jobs of the system with deadline $\leq \ttd_S(\tau_i)$.
        This quantity corresponds to the work that must be scheduled before $\ttd_S(\tau_i)$ without mode switch,
        so it is necessary, in a uniprocessor system, that this quantity does not exceed $\ttd_S(\tau_i)$.
    \end{proof}
\end{oracle}
The above result does not anticipate a mode switch (\ie{} $\alpha = \crit_S$),
it gives the necessary condition foreseeing that no mode switch occurs.
Using a similar construct but setting $\alpha = \hi{}$,
a necessary condition that foresees a mode switch can be derived as follows.

\begin{lemma}\label{lemma:nc-hi-dbf}
    For feasibility it is \emph{necessary} to have
    $
    \forall \tau_i \in \Active{\tau}:
    \ttd_S(\tau_i) \geq \dbf_{S}^{\hi{}}(\ttd_S(\tau_i))
    $.
\end{lemma}


This necessary condition guarantees that all $\hi$-tasks in $S$ will have enough time to execute all the $\rct$ of both their current and future jobs, with their deadline prior to those of any active tasks' jobs,
under the assumption that no job explicitly signals completion, eventually triggering a mode switch at the next possible instant.
%
%
Note that this necessary condition optimistically assumes that no future computing time would be wastefully given to a $\lo{}$ task until $\ttd_S(\tau_i)$.
In practice, should a $\lo$ task be scheduled, then $\dbf_{S}^{\hi{}}(\ttd_S(\tau_i))$ would remain constant, whereas $\ttd_S(\tau_i)$ will decrease, potentially violating this necessary condition.
%
%
%
%


\begin{oracle}[$\hi{}$ over demand]\label{oracle:hi-over-demand}
    $\{S \mid \exists \tau_i \in \Active{\tau}: \ttd_S(\tau_i) < \dbf_{S}^{\hi{}}(\ttd_S(\tau_i))\} \subseteq \Unsafe$.
    \begin{proof}
        If $\crit_S = \hi$, then the \autoref{lemma:nc-hi-dbf} is equivalent to \autoref{lemma:nc-dbf}, and the proof provided for \autoref{oracle:over-demand} holds.
        If $\crit_S = \lo$,
        we consider two cases,  one where a mode switch, triggered by a $\hi$-task $\tau_k \in \{\tau \mid L_k=\hi\}$, can occur before $\ttd_S(\tau_i)$, and the other where it may not.
        For a (current or future) job of $\tau_k$ to trigger a mode switch, it is sufficient to have $\rct(\tau_k)>0 \wedge \ttd_S(\tau_i)\geq\ttd_S(\tau_k)$ or $\rct(\tau_k)=0 \wedge \ttd_S(\tau_i)\geq\ttd_S(\tau_k)+T_k$.
        In such cases, we can anticipate the mode switch and enforce
        $\forall \tau_i \in \Active{\tau}: \ttd_{S'}(\tau_i) \geq \dbf_{S'}^{\crit_{S'}}(\ttd_{S'}(\tau_i))$ 
        as a necessary condition as per \autoref{lemma:nc-dbf}, with $S' = \critUp{S}{\bot}$.
        Indeed be noted that $\critUp{S}{\bot}$ merely increases the $\rct$ of the active $\hi$-jobs and discards the $\lo$ tasks.
        Therefore, if anything, this anticipation is more optimistic since it forestalls the wasteful scheduling of $\lo$-jobs.
        Furthermore, given $\crit_{S'}$=$\hi$, we can express $\dbf_{S'}^{\crit_{S'}}(\ttd_{S'}(\tau_i))$ as $\dbf_{S'}^{\hi{}}(\ttd_{S'}(\tau_i))$,
        so if this inequality does not hold, $S'$ is deemed unsafe and so is $S$ as $S'$ is more optimistic.
        If a mode switch occurrence cannot be guaranteed before $\ttd_S(\tau_i)$, \ie{} when $\rct(\tau_k)>0 \wedge \ttd_S(\tau_i)<\ttd_S(\tau_k)$ or $\rct(\tau_k)=0 \wedge \ttd_S(\tau_i)<\ttd_S(\tau_k)+T_k$, then $\df_{S}^{\hi}(\tau_k, \ttd_S(\tau_i))=0$ per definition.
        Moreover, for $\lo$-tasks $\tau_j \in \{\tau \mid L_j=\lo\}$, $\df_{S}^{\hi}(\tau_j, \ttd_S(\tau_i))=0$, and $\ttd_S(\tau_i) \geq 0 = \dbf_{S}^{\hi{}}(\ttd_S(\tau_i))$ per definitions. Thus, under these circumstances, this oracle's inequality never holds and $S$ is never deemed unsafe by it.

    \end{proof}
  \end{oracle}

  We can define a \emph{stronger} condition that generalises the negative laxity one. Indeed, the minimal laxity must be larger than $0$ (or null), the sum of the \emph{two} minimal laxities must be larger than $1$ (or $1$), etc.

\begin{definition}[Sum minimal laxities]
	In this definition, we assume that the laxities of active tasks are arranged in ascending order: $\laxity_S(\tau_1) \leq \laxity_S(\tau_2) \leq \cdots \leq \laxity_S(\tau_k)$. We define $\sumlax_k$ as the sum of the $k$ minimal laxities: $\sumlax_k = \sum_{i=1}^k \laxity_S(\tau_i)$.
\end{definition}

\begin{lemma}\label{lemma:minimal-laxities}
For the feasibility of $n$ sporadic tasks upon uniprocessor it is \emph{necessary} to have
$\forall 1 \leq k \leq n: \sumlax_k > k-2$.
\end{lemma}
\begin{proof}
    By induction.
    
    \emph{Base step}:
    When $k=1$, according to Lemma~\ref{lemma:nc-poslax}, it is necessary that the smallest task laxity ($\sumlax_1$) to be non negative.
    \emph{Inductive step}:
    We will show that the property also holds for $n=k+1$.
    By contradiction: we assume that $\sumlax_k = k - 2 + \epsilon > k-2$ ($\epsilon > 0$, induction hypothesis), and $\sumlax_{k+1} \leq k-1$ and the system is feasible.
    Notice that $\laxity_S(\tau_i)$ is an integer $\forall \tau_i \in \tau$, as we are considering discrete time systems (\ie{} $T_i$, $D_i$, etc. are all integers).
    \begin{align*}
        \sumlax_{k+1}   & = \sumlax_{k} + \laxity_S(\tau_{k+1}) \leq k-1                        & \text{by definition of $\sumlax_{k+1}$ and assumption} \\
                        & = k - 2 + \epsilon + \laxity_S(\tau_{k+1}) \leq k-1                   & \text{induction hypothesis} \\
                        & \Rightarrow \epsilon + \laxity_S(\tau_{k+1}) \leq 1                   & \text{rearranging terms} \\
                        & \Rightarrow \laxity_S(\tau_{k+1}) < 1                                 & \text{as $\epsilon > 0$} \\
                        & \Rightarrow  \laxity_S(\tau_{k+1}) \leq 0                             & \text{as $\laxity_S(\tau_{k+1})$ is an integer} \\
                        & \Rightarrow \forall 1 \leq i \leq k : \; \laxity_S(\tau_{i}) \leq  0  & \text{as tasks are sorted by laxity}
    \end{align*}
    Since $k > 1$ at least \emph{two} laxities are null simultaneously which contradicts the system feasibility.
\end{proof}

To our knowledge, this generalisation is original regarding the state of the art.

\begin{oracle}[Sum-min-laxity]
$\{S \mid \exists \tau_k \in \Active{\tau} :  \sumlax_k \leq k-2 \} \subseteq \Unsafe,$
assuming tasks $\tau_k \in \tau$ are sorted by laxity.

\begin{proof}
    Lemma~\ref{lemma:minimal-laxities}.
\end{proof}
\end{oracle}
Then, one can combine the concept of worst-laxity with that of sum of minimal ones to obtain an even stronger condition (at least in $\lo$ mode).

\begin{corollary}\label{corollary:minimal-laxities}
    $\forall \tau_i$ it is \emph{necessary} to have $\forall 1 \leq k \leq n: \sumlax'_k > k-2 \text{.}$

    where $\sumlax'_k = \sum_{i=1}^k \wlaxity_S(\tau_i)$, tasks are considered by non decreasing worst-laxity: $\sumlax'_1 \leq \cdots \leq \sumlax'_{n}$
\end{corollary}

\begin{proof}
    This immediately follows from Lemma~\ref{lemma:worst-laxity} and the proof of Lemma~\ref{lemma:minimal-laxities}.
\end{proof}

\begin{oracle}[Sum-min-worst-laxity]
$\{S \mid \exists \tau_k \in \Active{\tau} :  \sumlax'_k \leq k-2 \} \subseteq \Unsafe$,
assuming tasks $\tau_k \in \tau$ are sorted by worst-laxity.
\begin{proof}
    Corollary~\ref{corollary:minimal-laxities}
\end{proof}
\end{oracle}




\section{Evaluation}
\label{sec:evaluation}

To demonstrate the benefits of the search space reduction and the accuracy of the exact test, we developed a C++ tool~\cite{mcgraph_explorer} that takes a task set and scheduling policy as inputs, converts it into an automaton as in \autoref{sec:automaton-semantic}, and explores it according to \autoref{acbfs-safe}, allowing to assess the schedulability of the task set.
The same task set and scheduling policy can be explored with different optimisations (simulation relation and oracles), resulting in the same outcome (safe or unsafe) but with different time and space complexity.
%
%
Experiments were run on a server computer with 128 GB of RAM and a 128-core AMD Ryzen Threadripper 3990X CPU running at 2.9GHz.

\subsection{Random task set parameter generation}

Inspired by the works of Baruah~\cite{baruah2011response} and Ekberg~\cite{ekberg2014bounding}, the task sets used in our experiments were randomly generated using a python tool~\cite{mcgraph_explorer} as follows:
$T_i$ is an integer number generated according to a log-uniform distribution from the range of $[T^{\min}, T^{\max}]$;
$L_i = \hi$ with probability $0\leq P_\hi \leq 1$, otherwise $L_i = \lo$,
and  $D_i = T_i$.
The generation of $C_i$ values depends on a target average utilisation $U^* \in [0,1]$.
Then $\delta$ is a \texttt{float} drawn from the uniform distribution $\mathcal{U}\left(-\mu, \mu\right)$ with $\mu = \min\{U^*, 1-U^*\}$.
The target $\lo$-criticality utilisation of tasks, $U^{*\lo}$, are generated using the Dirichlet-Rescale (DRS) algorithm \cite{GriffinRTSS2020}, that gives an unbiased distribution of utilisation values.
We give the following inputs:
$n$ the number of tasks, $U= U^*+\delta$, $u_i^{\min} = 1/T_i$ and $u_i^{\max} = 1$ for $1 \leq i \leq n$.
Similarly, the target $\hi$-criticality utilisations of tasks, $U^{*\hi}$, are generated using DRS too with, as inputs, the same number of tasks $n$, $U=U^*-\delta$, $u_i^{\min} = U^{*\lo}(\tau_i)$ and $u_i^{\max} = 1$ if $L_i = \hi$, $u_i^{\min}=0=u_i^{\max}$ otherwise, for $1 \leq i\leq n$.
$C_i(\lo) = U^{*\lo}(\tau_i)\cdot T_i$ rounded to the nearest integer.
$C_i(\hi) = U^{*\hi}(\tau_i)\cdot T_i$ rounded to the nearest integer if $L_i = \hi$ and $C_i(\hi)=C_i(\lo)$ otherwise.
%
%
Task sets were dropped if $U^\lo(\tau) > 1$ or $U^\hi(\tau)> 1$.
Duplicate task sets were discarded, as were task sets with all tasks sharing the same criticality level.
Additionally, task sets were dropped if $\mid U^{\avg}(\tau)-U^*\mid > 0.005$.
%
%
%
Deadlines are implicit because our explorations will use among others EDF-VD (see \autoref{edfvd}) which can only cope with such deadlines~\cite{BaruahBDMSS11}.
Hereunder, a range of numbers from $f$ to $t$ with a step increase of $s$ will be denoted by $[f;t;s]$.
\subsection{Antichain impact on state space exploration}


\begin{figure}
    \centering
    \includegraphics[width=\columnwidth]{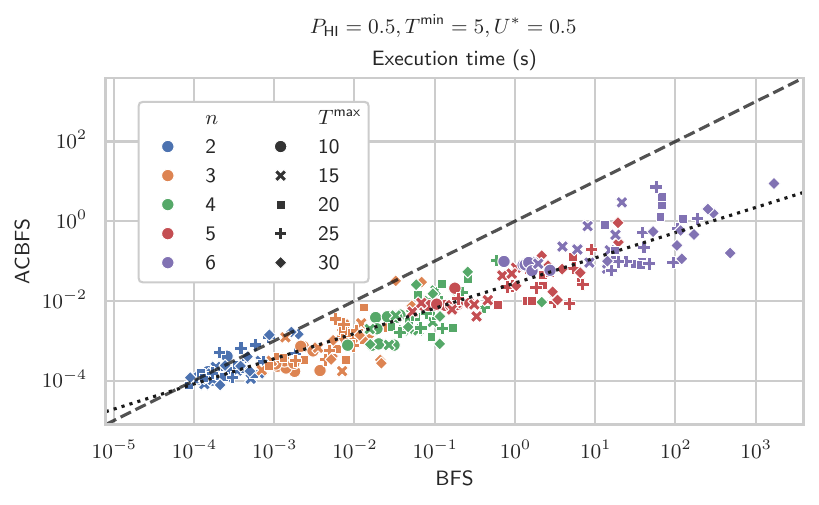}
    \caption{%
        Execution times in seconds before halt for BFS and ACBFS. Dashed line is $x=y$, dotted line is the linear regression on the samples (after $\fct{log}_{10}$).
    }
    \label{fig:state-bfs-acbfs}
\end{figure}

\autoref{fig:state-bfs-acbfs} compares the performance of antichain breadth first search using the idle tasks simulation relation (ACBFS) and classical breadth first search (BFS) with EDF-VD and without any  oracle.
Task sets were generated with varying number of tasks and maximum periods, the rest of the parameters being fixed.
For each combination of parameters, 10 sets were generated.
All task sets were schedulable, hence, their automaton was fully developed as every single state reachable from the initial state was visited (or simulated).
ACBFS outperforms BFS in execution time, except for the few cases where the search space is very small
--- in such cases, both algorithms ran for less than $0.01$ second.
This shows that ACBFS and the idle tasks simulation relation scale better than BFS, reducing the exploration time by at least one order of magnitude.
As the state space size increases, the execution time gap increases, as indicated by the regression line.
Results on the number of states visited for each of the algorithms are on par, but analysing their execution time takes into account the (potential) extra computing time of the antichains.



\begin{figure}
    \centering
    \subfigure[]{
        \includegraphics[width=0.49\columnwidth]{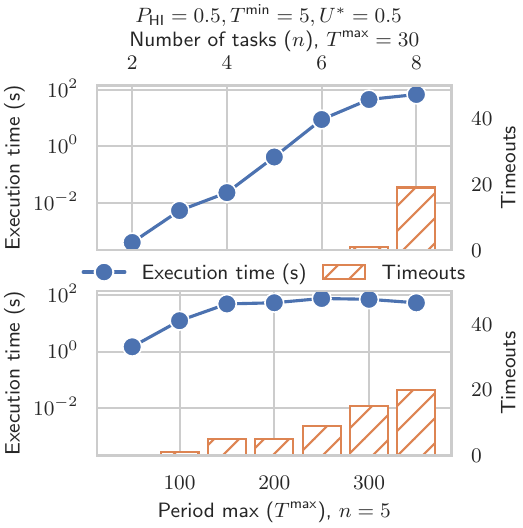}
        \label{fig:acbfs-scalability}
    }%
    \subfigure[]{
        \includegraphics[width=0.49\columnwidth]{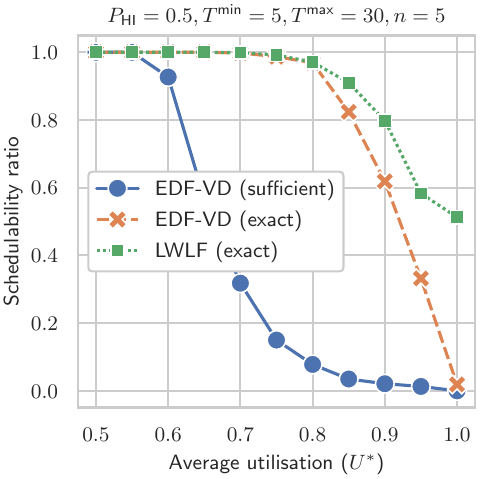}
        \label{fig:scheduling}
    } 
    \caption{\subref{fig:acbfs-scalability} ACBFS' median execution times over 40 explorations before halt, with varying number of task and maximum period and a \timeoutMins{} minutes timeout.\\\subref{fig:scheduling} Schedulability ratio over average utilisation per test. Exploration with ACBFS and $\hi$ over demand for exact tests.}
    \label{fig:foobar}
\end{figure}


\autoref{fig:acbfs-scalability} illustrates the scalability of ACBFS with a varying quantity of tasks and maximum period.
Both parameters contribute to prolonging the exploration time, but the number of tasks exerts a stronger and more consistent effect.
The automaton's size being on the order of $\mathcal{O}({(T^\max)}^n)$ explains this stronger impact for the number of tasks.
The period's effect is less stable, given that the $T^\max$ parameter serves as an upper bound to the distribution from which the period is drawn; as such, even with a high $T^\max$ value, the periods can still be set at lower levels.
The majority of task sets, composed of 8 tasks, were successfully explored within the allocated \timeoutMins-minute timeout.
As for the maximum period parameter, timeouts were observed as early as $T^\max=100$.
Nearly a third of the task sets with $T^\max=350$ exceeded the timeout, suggesting that specific combinations of periods can lead to a larger state space.

While the antichain optimization enables a reduction in the state space, the problem remains challenging and experiences exponential growth, especially considering that it deals with sporadic tasks.
Most of the task sets with $n=7$ and some with $n=8$ can be explored within an acceptable timeframe, as ACBFS can manage at least one more task than BFS (see \autoref{fig:state-bfs-acbfs}).
Beyond, both the execution time and required amount of memory (the maximal explored states are stored as per \autoref{acbfs-safe}) to explore some of the task sets become prohibitive.
In comparison, other papers providing an exact test for (single-criticality) \emph{sporadic} tasks sets with FJP schedulers experimented with task sets parameters up to $T^\max=6$~\cite{bakerbrute}, $n=8,T^\max=8$~\cite{geeraerts2013multiprocessor} and $n=4,T \in \{1, 2, 5, 10, 20, 50, 100, 200, 1000\}$~\cite{yalcinkaya2019exact}.
In real-world situations, certifying a task set requires a single exploration;
in those cases, the exploration time could extend longer than in our experiments, allowing to increase task number or periods.

\subsection{Oracle impact}



To measure the impact of the oracles, $2100$ task sets were generated with $T^{\min} = 5, T^{\max} = 30, P_\hi = 0.5, n = 5$ and $U^* \in [0.8;1;0.01]$ with $100$ task sets for each values of $U^*$.
\autoref{fig:acbfs-scalability} shows our simulator is capable of accommodating larger systems.
The parameters here selected expedite the exploration process, thereby enhancing the quantity of results and facilitating the accrual of more robust, aggregated metrics.
%
%
All task sets were explored with ACBFS and EDF-VD, several times with different oracles.
As baseline, an exploration with ACBFS without any oracle is used.
During those explorations, 923 task sets were found to be unschedulable.
Thus, about half of the task sets are schedulable, obtained over several (omitted) iterations of the experiments to tune the range of utilisations.
We expect the impacts of oracles to differ depending on the actual schedulability of the evaluated task sets.
For example, unsafe oracles computed on schedulable task sets will not reduce the search space and will only constitute overheads
--- this, however, cannot be known a priori by the system designer.
\begin{figure}[t]
    \centering
    \includegraphics[width=\columnwidth]{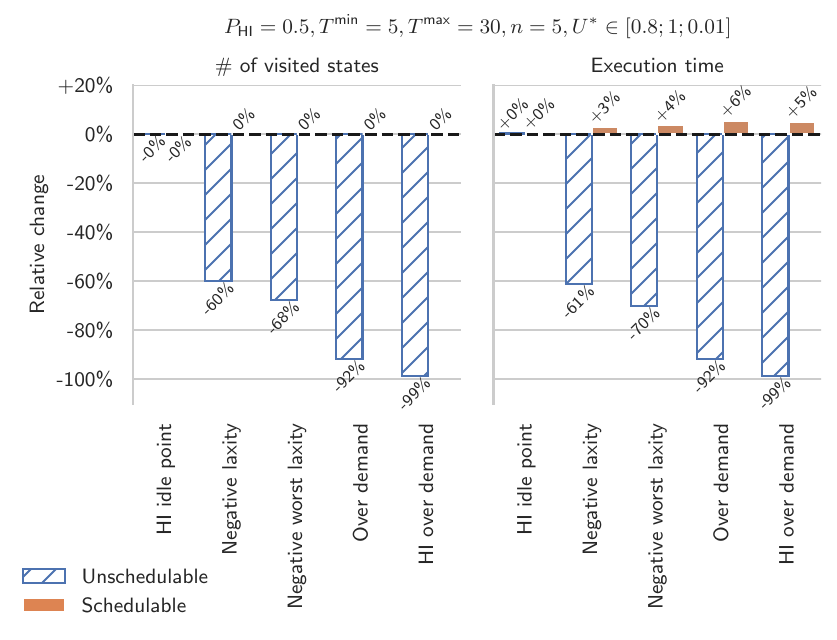}
    \caption{%
        Oracles' impact on the number of states visited and the execution time.
    }
    \label{fig:oracle-impact}
\end{figure}
For each oracle, \autoref{fig:oracle-impact} shows the avoided state ratio and speed-up ratio between ACBFS used with and without the oracle.
The reported values are median over all explorations, split over task sets schedulability.

The $\hi$ idle point safe oracle barely reduces the number of visited states for schedulable and unschedulable task sets (0.02\% and 0.07\%).
%
The execution time increases are both 0.4\%, reflecting the additional computing cost of the oracle.
%
%
%
%

The first unsafe oracle, negative laxity, already allows avoiding 60.1\% of the states, leading to an execution time reduction of 61.2\% for unschedulable task set, while increasing the execution time by 3.0\% of the schedulable ones
%
%
--- computing the laxity
only requires $\mathcal{O}(n)$ operations
on each visited state.
%

Then, the negative worst laxity unsafe oracle builds on top of the laxity and brings another level of impact.
67.7\% of the states are dropped, translating to a time execution reduction of 70.1\% for unschedulable task sets, while only having a median time execution increase of 3.6\% on schedulable task sets.
This significant impact comes from the anticipation of what will happen in $\hi$ mode when still in $\lo$ mode,
hence enabling to detect a deadline miss \emph{earlier}.
%

The over demand oracle computes the current mode demand.
%
This approach brings a step change as it avoids 91.6\% of the states and reduces the execution time by 91.7\% for unschedulable task sets, while increasing the execution time of the schedulable task sets by 5.6\%.
While strong, it is more expensive than the other unsafe oracles, as shown by the increased time with schedulable task sets.

Finally, the $\hi$ over demand oracle combines both worlds by analysing all tasks and anticipating the demand at the next criticality level.
This approach is the strongest, dropping 98.8\% of states, reducing the execution time by 98.6\% for unschedulable task sets and increasing the execution time by 5.2\% on schedulable task sets.

\begin{table}
\centering
\caption{%
    Statistics on the number of visited states.
}
\label{tab:explo-metrics}
\begin{tabular}{lllll}
\multicolumn{5}{c}{$P_\hi = 0.5, T^\min=5, T^\max=20, n=5, U^*\in [0.8;1;0.01]$}\\
\toprule
Search & \multicolumn{2}{l}{BFS} & \multicolumn{2}{l}{ACBFS} \\
Oracle & None & HI over demand & None & HI over demand \\
\midrule
min & 9905 & 35 (-100\%) & 668 (-93\%) & 29 (-100\%) \\
mean & 746974 & 480688 (-36\%) & 75374 (-90\%) & 46024 (-94\%) \\
std & 999984 & 819378 (-18\%) & 126110 (-87\%) & 107731 (-89\%) \\
median & 410063 & 217928 (-49\%) & 35888 (-91\%) & 15459 (-96\%) \\
max & 11875126 & 11875126 (-0\%) & 2687577 (-77\%) & 2687577 (-77\%) \\
\bottomrule
\end{tabular}

\end{table}

To summarise the combined impacts of antichains with idle tasks simulation and the best oracle ($\hi{}$ over demand),
\autoref{tab:explo-metrics} outlines statistics.
Generation parameters are the same as the oracle impact experiment, \autoref{fig:oracle-impact}, except for $T^\max=20$, resulting in balanced schedulability.
The table shows statistics about the composed effect of both the antichains and the best oracle,
achieving up to \speedup{}\% reduction in the searched state space (2968037 visited states for BFS against 42 for ACBFS with the oracle).

\subsection{Scheduling analysis}


Our framework enabled us to step into scheduling experimentations,
aiming to understand the real performance of EDF-VD.
Specifically, automata were explored with ACBFS, the $\hi$ over demand unsafe oracle and without safe oracles.
Task sets were generated with $T^{\min} = 5, T^{\max} = 30, P_\hi = 0.5, n = 5$ and $U^* \in [0.5;1;0.05]$.
For each utilisation, $1000$ task sets were generated.
\autoref{fig:scheduling} shows the average schedulability ratio for each target average utilisation.
It shows the schedulability ratio of EDF-VD's sufficient test~\cite{BaruahBDMSS11},
EDF-VD exact test using our automaton exploration approach,
and of a novel algorithm \textit{Least Worst Laxity First} (LWLF).

\begin{definition}[LWLF scheduler]
  The Least Worst Laxity First scheduler, or LWLF, is defined as follows.
  We let $\MinWL(S)$ be the task $\tau_i\in\Active{S}$ which has the minimal
  worst laxity in $S$.
  That is, $\tau_i$ is s.t. for all $\tau_k\in \Active{S}\setminus\{\tau_i\}$:
  $\wlaxity_S(\tau_k) > \wlaxity_S(\tau_i)$ or $\wlaxity_S(\tau_k) = \wlaxity_S(\tau_i) \wedge k > i$.
  Then, for all states $S$, we let:
  \begin{numcases}
    {\lwlf(S)=}
      \bot & if $\Active{S}=\emptyset$ \nonumber \\
      \MinWL(S)  & otherwise. \nonumber
  \end{numcases}
\end{definition}

EDF-VD's actual schedulability significantly outperforms its sufficient test, underscoring the test's pessimism.
Notice, however, that this exact test --- even with antichains and oracles --- can be a time-intensive process, whereas the EDF-VD's sufficient test is computed in polynomial time.
For large values of $n$,  the exact test might be impractical, so sufficient tests remain relevant.
One can always first run a sufficient test, and if negative, try the exact analysis.

LWLF brings further schedulability significant improvements over EDF-VD.
We derived LWLF after observing the strong impact of the worst laxity oracle
---
since it can detect a failure earlier, it hinted the existence of a scheduling algorithm integrating this information to guide the scheduling decision.
Similarly, we believe that a scheduling algorithm building on top of the $\hi$ over demand oracle can be derived and is expected to bring even better results.

\section{Related Work}

There is extensive related research on real-time scheduling.
Below, we categorize and summarize the most relevant prior work to ours.

\textbf{Scheduling tests.}
Utilisation based conditions were produced, both sufficient~\cite{BaruahBDMSS11} and necessary~\cite{santos2015considerations}, the latter also adapting them for multiprocessor.
Demand-bound function sufficient schedulability tests have also been produced, for dual-criticality~\cite{ekberg2012outstanding,chen2014efficient},
then extended to any number of criticalities~\cite{gu2017efficient} and to cope with deadline reduction on mode change~\cite{ekberg2014bounding}.

\textbf{Formal verification for mixed-criticality scheduling.}
Burns and Jones used time bands and rely-guarantee conditions~\cite{jones2020rely, burns2022approach}.
Inspired by Hoare’s logic, they proposed a framework to formally specifying the temporal behaviour of schedulers, 
hence allowing the verification of schedulers and their code generation.
Abdedda\"im~\cite{yasmina2020accurate} provided exact schedulability tests for FJP schedulers
for a more modular mixed criticality model, where only specified subsets of $\lo$ (sporadic) tasks are discarded in $\hi$ mode depending on the subset of $\hi$ tasks exceeding their $\lo$ budget.

\textbf{Using automata to check schedulability.}
\iftrue
Already discussed in our introduction (see~\autoref{sec:introduction}).
\else
As detailed in~\autoref{sec:introduction}, developing an exact scheduling test via automaton exploration was introduced by Baker and Cirinei~\cite{bakerbrute}
for the case of single-criticality sporadic tasks on multi-processor platforms under any scheduling algorithm, including the dynamic priority class.
The approach was extended to multiprocessor scheduling and with the use of antichains in the works of Lindström~\etal~\cite{lindstrom2011faster,geeraerts2013multiprocessor}.
Their study incorporated the antichain technique to improve the performance of Baker's state enumeration algorithm~\cite{bakerbrute}, which suffer from state space explosion.
They also extended the model to arbitrary deadlines.
By comparison, we also implemented antichains
but for uniprocessor \emph{dual-criticality} systems,
complemented by further reduction of the search space by leveraging \emph{oracles}.
Asyaban~\etal~\cite{asyaban2018exact} released an exact schedulability test for uniprocessor dual-criticality systems for FTP schedulers, leveraging the idea of pruning the search space with specific insights.
By contrast, we support any scheduling algorithm (including from the \emph{DP} class).
\fi

\textbf{Antichains.}
The \emph{antichain technique} has been introduced by De~Wulf~\etal~\cite{DDHR06} to improve the practical performances of costly algorithms on automata.
These ideas have later been applied to produce a series of efficient algorithms for automata~\cite{DR10,ACHMV10}, model-checking~\cite{DR09} and games~\cite{RCDH07,FJR11,GGNS18}.

\textbf{Schedule-Abstraction Graphs (SAG).}
Another reachability-based response-time analysis was developed for an exact test on sets of non-preemptive jobs with release jitter and variable execution time with a FJP scheduler upon a uniprocessor~\cite{nasri2017sag}, then upon multiprocessor~\cite{nasri_et_al:LIPIcs.ECRTS.2018.9}, with preemption~\cite{gohari_et_al:LIPIcs.ECRTS.2024.3}, and several kinds of inter-job constraints~\cite{nasri_et_al:LIPIcs.ECRTS.2019.21,9355543,nelissen_et_al:LIPIcs.ECRTS.2022.12}.
Applying partial-order reduction rules into SAG mitigated
combinatorial explosion~\cite{9804591,Ranjha2023}.
Notice a sequence of jobs
--- the expected input of SAG methods ---
\emph{cannot support} the sporadic task model.
By definition, the set of jobs is not known in advance,
which is not the case in systems
that are periodic or with release jitters.

\section{Conclusion}

We have developed a generic framework for exact schedulability assessment in uniprocessor mixed-criticality systems,
by reducing it to the safety problem in an automaton.
We combined insights from the formal verification community
--- like antichains, simulation relations and the ACBFS algorithm ---
and from the real-time research
--- like safe and unsafe oracles ---
to make the approach more practical.
We demonstrated in simulations that those allow to reduce the search space by up to \speedup{}\%,
enabling the evaluation of existing and new scheduling algorithms, EDF-VD and LWLF,
and comparing it to a prior sufficient test~\cite{BaruahBDMSS11}.

\textbf{Future work.}
Although the scalability of our approach remain limited (up to 8 tasks),
our algorithm is generic,
allowing one to incorporate a range of other optimisations,
and can be applied to other scheduling problems.
%
%
Our work can be extended to consider \emph{more realistic} mixed criticality models~\cite{burns2022mixed},
\eg regarding $\lo$-criticality tasks that should not be abandoned callously.
We then need to modify the completion transition, in order to
let $\lo$-tasks run until completion
or adjust the period and deadline of certain $\lo$-tasks.
A \emph{multi-processor} version of the model is already supported in prior work for sequential jobs~\cite{lindstrom2011faster,geeraerts2013multiprocessor}, we would have to combine it with the extensions proposed in this work.
In addition, the reduction can be adapted to support other platforms and task models
such as modelling \emph{varying CPU speed} and \emph{job-level parallelism} (\eg the gang model),
by reducing the $\rct$ values according to the CPU speed and the number of CPU cores assigned to the job.
The $\schedule(S)$ function would return the selected CPU speed factor together with how many CPUs are assigned to each job.
A power function could then be used to label the clock-tick transitions with numeric values,
hinting shortest path exploration (instead of breadth-first search) to minimize power consumption.
%
%
To account for \emph{preemption delays}, the model could force idle CPU time by forcing $\run=\bot$ for a given number of run transitions when the scheduler is switching to a different task across calls.

\bibliographystyle{alpha}
\bibliography{bibliography}
\clearpage

\appendix

\section{Proof of 
  Proposition~\ref{prop:algo-correct-and-terminates}}\label{app:proof-algo}

In the present section, we prove the correcteness and termination of
our antichain- and oracle-enhanced algorithm. For the sake of
readability, we recall the proposition here:

\textsc{Proposition \ref{prop:algo-correct-and-terminates}.}
On all automata $A$, \autoref{acbfs-safe} terminates and returns
`\texttt{Fail}' iff $A$ is unsafe.

For the sake of readability of this section, we recall that
\autoref{acbfs-safe} relies on the computation of two sequences
$\left(\Rtilde_i\right)_{i\geq 0}$ and
$\left(\Ntilde_i\right)_{i\geq 0}$ defined as follows:

\begin{align*}
  \Ntilde_0 = \Rtilde_0 = & \{v_0\}\setminus\dc{\Safe} \\
  \forall i\geq 0: \Ntilde_{i+1} = & \Max{\Succ{\Ntilde_i}\setminus
  \dc{\Rtilde_i\cup \Safe}} \\
  \forall i\geq 0: \Rtilde_{i+1} = & \Max{\Rtilde_i\cup \Ntilde_{i+1}}.
\end{align*}

We proceed with the proof by introducing some ancillary results.  We
first prove some easy observations regarding the two sequences
$\left(\Rtilde_i\right)_{i\geq 0}$ and
$\left(\Ntilde_i\right)_{i\geq 0}$.
\begin{lemma}\label{lemma:obs-seq}
  The following holds on
  $\left(\Rtilde_i\right)_{i\geq 0}$ and
  $\left(\Ntilde_i\right)_{i\geq 0}$:
  \begin{enumerate}
  \item\label{item:1} for all $i\geq 0$:
    $\Ntilde_i\subseteq\Reach{v_0}$;
  \item\label{item:2} for all $i\geq 0$, for all $v\in \Rtilde_i$:
    there exists $0\leq j\leq i$ s.t. $v\in\Ntilde_j$;
  \item\label{item:3} for all $i\geq 0$:
    $\Rtilde_i\subseteq\Reach{v_0}$;
  \item\label{item:4} for all $i\geq 0$: $\Ntilde_i=\emptyset$ implies
    that $\Ntilde_j=\emptyset$ for all $j\geq i$;
  \item\label{item:5} for all $i\geq 0$:
    $\dc{\Rtilde_i}\subseteq\dc{\Rtilde_{i+1}}$.
  \end{enumerate}
\end{lemma}
\begin{proof}\leavevmode

	\begin{enumerate}
		\item This point stems from the fact that
		      $\Ntilde_0\subseteq \{v_0\}\subseteq \Reach{v_0}$, and that, for all
		      $i\geq 1$, $\Ntilde_i\subseteq\Succ{\Ntilde_{i-1}}$.
		\item This point is easily obtained by observing the
		      definitions. For $i=0$, we have $\Ntilde_0=\Rtilde_0$. Then, all
		      sets $\Rtilde_i$ (with $i\geq 1$) are obtained by adding elements
		      from $N_i$ (and removing non-maximal elements);
		\item This point stems from the fact that
		      $\Rtilde_0\subseteq \{v_0\}\subseteq \Reach{v_0}$ and that all sets
		      $\Rtilde_i$ (for $i\geq 1$) contain only elements from the sets
		      $\Ntilde_j$ (with $j\leq i$), which are all reachable from $v_0$
		      by point~\ref{item:1}.

		\item This point stems from the fact that
		      $\Succ{\emptyset}=\emptyset$. Since
		      $\Ntilde_{i+1}\subseteq \Succ{\Ntilde_i}$ for all $i\geq 0$, the
		      fact that $\Ntilde_k=\emptyset$ implies that
		      $\Ntilde_{k+1}=\emptyset$ as well, and the argument carries on
		      inductively to show that all $\Ntilde_{j}$ with $j\geq k$ are
		      empty as well.
		\item Finally, this last point stems from the definition. For all
		      $i\geq 0$, the following holds:
		      \begin{align*}
			      \dc{\Rtilde_{i+1}} & = \dc{\Max{\Rtilde_i\cup \Ntilde_{i+1}}} \\
			                         & = \dc{\Rtilde_i\cup \Ntilde_{i+1}}       \\
			                         & = \dc{\Rtilde_i}\cup \dc{\Ntilde_{i+1}}.
		      \end{align*}
		      Hence, $\dc{\Rtilde_i}\subseteq \dc{\Rtilde_{i+1}}$ for all
		      $i\geq 0$.
	\end{enumerate}
\end{proof}

While points~\ref{item:1} and \ref{item:3} above show that the
sequence computes only reachable elements, we also need to show that
the agressive optimisation introduced by the antichains, and the
$\Safe$ and $\Unsafe$ sets do not `cut too many states'. This is the
point of the next lemma:
\begin{lemma}\label{lem:path-anti}
	Let $v_0,v_1,\ldots, v_k$ be a path s.t., for all $0\leq i\leq k$:
	$v_i\not\in\dc{\Safe}$ and $v_i\not\in\uc{\Unsafe}$. Then, for all
	$0\leq i\leq k$: $v_i\in\dc{\Rtilde_i}$.
\end{lemma}
\begin{proof}
	The proof is by induction on the length of the path.

	For the base case ($k=0$), the lemma is trivial.

	For the inductive case, we consider a path $v_0,v_1,\ldots, v_k$
	that satisfies the conditions of the lemma, and we assume (induction
	hypothesis) that $v_{k-1}\in \dc{\Rtilde_{k-1}}$. Let us prove that
	$v_k \in \dc{\Rtilde_{k}}$.  First, we let $\ell\leq k-1$ be
	s.t. $v_{k-1}\in\Ntilde_\ell$; such an $\ell$ exists by
	\autoref{lemma:obs-seq}, point~\ref{item:2}, since
	$v_{k-1}\in \dc{\Rtilde_{k-1}}$. By definition of the
	$\left(\Ntilde_i\right)_{i\geq 0}$ sequence, we know that:
	\begin{align*}
		\Ntilde_{\ell+1} & =\max{\Succ{\Ntilde_\ell}\setminus\dc{\Rtilde_\ell\cup\Safe}}.
	\end{align*}
	Since we consider a path $v_0,v_1,\ldots, v_k$, we know that
	$v_k\in\Succ{v_{k-1}}\subseteq\Succ{\Ntilde_\ell}$. We consider two
	cases:
	\begin{enumerate}
		\item either $v_k\in\dc{\Rtilde_\ell}$. In this case, by
		      \autoref{lemma:obs-seq}, point~\ref{item:5}:
		      $v_k\in\dc{\Rtilde_k}$, since $\ell\leq k$;
		\item or $v_k\not\in\dc{\Rtilde_k}$. In this case, since
		      $v_k\not\in\Safe$, by hypothesis, we conclude that
		      $v_k\in\Ntilde_{\ell+1}$, by the above definition of
		      $\Ntilde_{\ell+1}$. By definition of the
		      $\left(\Rtilde_i\right)_{i\geq 0}$, this implies that
		      $v_k\in\dc{\Rtilde_{\ell+1}}$. Hence, $v_k\in\dc{\Rtilde_k}$ as
		      well, by \autoref{lemma:obs-seq}, point~\ref{item:5}, as
		      $\ell\leq k-1$.
	\end{enumerate}
	In both cases, we conclude that  $v_k\in\dc{\Rtilde_k}$. Hence the Lemma.
\end{proof}

We are now ready to prove the algorithm.

{\it Proposition \ref{prop:algo-correct-and-terminates}:} On all
automata $A$, Algorithm~\ref{acbfs-safe} terminates and returns
`\texttt{Fail}' iff $\Reach{A}\cap F\neq \emptyset$.
\begin{proof}
We first establish termination. It stems from points~\ref{item:3}
	and~\ref{item:5} of Lemma~\ref{lemma:obs-seq}:
	\begin{enumerate}
		\item for all $i\geq 0$,
		$\dc{\Rtilde_i}\subseteq\dc{\Rtilde_{i+1}}$;
		\item $\Rtilde_i\subseteq\Reach{A}$, hence
		$\dc{\Rtilde_i}\subseteq\dc{\Reach{A}}$; and
		\item $\Reach{A}$ is a finite set, hence $\dc{\Reach{A}}$ is finite
		too.
	\end{enumerate}
	Thus, we cannot have an infinite $\subset$-increasing chain of
	$\dc{\Rtilde_i}$. Thus, let $\ell$ be the least index
	s.t. $\dc{\Rtilde_\ell}=\dc{\Rtilde_{\ell+1}}$. As all $\Rtilde_i$
	contain only $\preceq$-incomparable elements by definition, this
	implies that $\Rtilde_\ell=\Rtilde_{\ell+1}$. Since
	$\Rtilde_{\ell+1}=\Max{\Rtilde_\ell\cup \Ntilde_{\ell+1}}$, we
	conclude that $\Ntilde_{\ell+1}\subseteq\dc{\Rtilde_\ell}$. However,
	by definition $\Ntilde_{\ell+1}\cap\dc{\Rtilde_\ell}=\emptyset$. The
	only way to satisfy these constraints is that
	$\Ntilde_{\ell+1}=\emptyset$. Hence, the algorithm terminates.
	\medskip

	Next, we establish soundness, i.e., when the algorithm returns
	`\texttt{Fail}', we have indeed
	$\Reach{A}\cap F\neq \emptyset$. The algorithm returns
	`\texttt{Fail}' only if
	$\dc{\Ntilde_i}\cap\uc{\Unsafe}\neq\emptyset$ for some $i$. This
	implies that $\Ntilde_i\cap\uc{\Unsafe}\neq\emptyset$. However, by
	Lemma~\ref{lemma:obs-seq}, point~\ref{item:1}, this implies that
	$\Reach{A}=\Reach{v_0}\cap \uc{\Unsafe}\neq\emptyset$, or, in other
	words, that we have found at least one reachable element
	$v\in\Reach{A}$ which is $\preceq$-larger than some element from
	$\Unsafe$. As such, $v$ is also an unsafe sate. Thus,
	$\Reach{A}\cap F\neq \emptyset$ \medskip

	Finally, we prove completeness of the algorithm, i.e., when
	$\Reach{A}\cap F\neq \emptyset$, the algorithm returns
	`\texttt{Fail}'. Since $\Reach{A}\cap F\neq \emptyset$, let
	$\rho=v_0, v_1,\ldots v_k$ be a path s.t.
	$v_k\in F$. Let $\ell$ be the least index
	s.t. $v_\ell\in\uc{\Unsafe}$. Such an index necessarily exists since
	$F\subseteq \uc{\Unsafe}$. Since $\rho$ reaches $F$, we also know
	that $v_i\not\in\Safe$ for all $0\leq i\leq k$. First, in the case
	where $\ell=0$, i.e. $v_0\in\uc{\Unsafe}$, then
	$\dc{\Ntilde_0}\cap\uc{\Unsafe}\neq\emptyset$ since $v_0$ is
        the initial state,
	and the algorithm returns `\texttt{Fail}' during the first
	iteration of the \textbf{repeat} loop. Otherwise, we need to show
	that $\dc{\Ntilde_K}\cap\uc{\Unsafe}\neq\emptyset$ for some $K$, and
	that the algorithm will indeed compute enough iterations of the
	\textbf{repeat} loop to eventually return `\texttt{Fail}'.

	We first show that $\dc{\Ntilde_K}\cap\uc{\Unsafe}\neq\emptyset$ for
	some $K$. To this end, we apply Lemma~\ref{lem:path-anti} to
	$v_0, v_1,\ldots, v_{\ell-1}$, which satisfies the conditions of the
	lemma. We obtain that $v_{\ell-1}\in \dc{\Rtilde_{\ell-1}}$. Let
	$v'_{\ell-1}\in\Rtilde_{\ell-1}$ be s.t.
	$v_{\ell-1}\preceq v'_{\ell-1}$, and let $i$ be the index
	s.t. $v'_{\ell-1}\in \Ntilde_i$. Such an $i$ exists by
	Lemma~\ref{lemma:obs-seq}, point~\ref{item:2}. Further, let
	$v'_\ell\in\Succ{v'_{\ell-1}}$ be a successor of $v'_\ell$
	s.t. $v_\ell\preceq v_{\ell-1}$. Such a $v'_\ell$ exists by
	$\preceq$-monotonicity of the automaton. Observe that
	$v'_\ell\in\Unsafe$ since $v_\ell\in\Unsafe$. Furthermore, since
	$v'_\ell\in\Succ{v'_{\ell-1}}$, with $v'_{\ell-1}\in \Ntilde_i$, we
	have $v'_\ell\in \Succ{\Ntilde_i}$. We consider two cases. Either
	$v'_\ell \in \Ntilde_{i+1}$. In this case, we obtain immediately
	that $\dc{\Ntilde_{i+1}}\cap\uc{\Unsafe}\neq\emptyset$. Or
	$v'_\ell \not\in \Ntilde_{i+1}$. By definition of $\Ntilde_{i+1}$,
	this can happen only because $v'_\ell\in\dc{\Rtilde_i}$, since
	$v'_\ell\not\in\Safe$. From $v'_\ell\in\dc{\Rtilde_i}$, we deduce
	that there is $v''_\ell\in\Rtilde_i$ s.t. $v'_\ell\preceq
		v''_\ell$. Again, this implies that $ v''_\ell\in\Unsafe$, since
	$v'_\ell\in\Unsafe$. By Lemma~\ref{lemma:obs-seq},
	point~\ref{item:2}, there exists $j$ s.t. $v''_\ell$ is also in
	$\Ntilde_j$. Thus, $v''_\ell\in \dc{\Ntilde_j}\cap\uc{\Unsafe}$. In
	both cases, we conclude that
	$\dc{\Ntilde_K}\cap\uc{\Unsafe}\neq\emptyset$ for some $K$.

	To conclude the proof, it remains to show that the algorithm will
	eventually compute enough iterations. Since $\Ntilde_K\neq\empty$,
	we conclude, by Lemma~\ref{lemma:obs-seq}, point~\ref{item:4} that
	$\Ntilde_i\neq\emptyset$ for all $0\leq i\leq K$. Thus, there are
	two possibilities. Either the algorithm terminates before testing
	whether $\dc{\Ntilde_K}\cap\uc{\Unsafe}\neq\emptyset$. This can only
	happen because the algorithm has returned `\texttt{Fail}',
	since the condition of the \textbf{until} will not be fulfilled
	before. Or the algorithm gets to the point where it tests
	$\dc{\Ntilde_K}\cap\uc{\Unsafe}\neq\emptyset$, and it terminates by
	returning `\texttt{Fail}'.
\end{proof}


\section{Proof of  \autoref{theo:idle-sim}}\label{app:proof-simu}

\quad\textsc{\autoref{theo:idle-sim}.}
  Let $\tau$ be a dual-criticality sporadic task system and $\schedule$ a deterministic
  uniprocessor scheduler for $\tau$.
  Then $\idle$ is a simulation relation for $A(\tau, \schedule)$.

\begin{proof}

  Let $S^\sreq_1$ and $S^\srun_1 \in \States{\tau}$, $\tau^\sreq \subseteq \tau $ and $\theta \in \lbrace\true,\false\rbrace$ be s.t. : $S_1 \RqTrans{\tau^\sreq} S^\sreq_1
  \CtTrans{\schedule(S^\sreq_1)} S_1^\srun \CpTrans{\schedule(S^\sreq_1),\theta} S_1^\scom$. 
  Those exist by \autoref{auto} and since $(S_1, S_1^\scom) \in E$.
  Let $S_2^\sreq$ be the (unique) state \sth{} $S_2 \RqTrans{\tau^\sreq} S_2^\sreq$ and let
  us show that $S_1^\sreq \idle S_2^\sreq $:

  \begin{enumerate}

    \item By \autoref{treq}: $\crit_{S_1^\sreq} = \crit_{S_1}, \crit_{S_2^\sreq} =
      \crit_{S_2}$.
      Since $S_1 \idle S_2 $, we also have that $\crit_{S_2} = \crit_{S_1}$.
      We conclude that $\crit_{S_2^\sreq} = \crit_{S_1^\sreq}$.

    \item Let $\tau_i \in \tau^\sreq : \rct_{S_1^\sreq}(\tau_i) = C(\crit_{S_1})$ and
      $\rct_{S_2^\sreq}(\tau_i) = C(\crit_{S_2})$.
      Since $S_1 \idle S_2 $, we know that $\crit_{S_2}(\tau_i) = \crit_{S_1}(\tau_i)$.
      Hence, $\rct_{S_2^\sreq}(\tau_i) = \rct_{S_1^\sreq}(\tau_i)$.
      For a task $\tau_i \notin \tau^\sreq$, we have $\rct_{S_1^\sreq}(\tau_i) =
      \rct_{S_1}(\tau_i)$ and $\rct_{S_2^\sreq}(\tau_i) = \rct_{S_2}(\tau_i)$.
      Since $S_1 \idle S_2 $, we know that $\rct_{S_2}(\tau_i) = \rct_{S_1}(\tau_i)$.
      Hence, $\rct_{S_2^\sreq}(\tau_i) = \rct_{S_1^\sreq}(\tau_i)$.
      We conclude that $\rct_{S_2^\sreq} = \rct_{S_1^\sreq}$.

    \item Let $\tau_i$ be \sth{} $\rct_{S_1^\sreq}(\tau_i) = 0$.
      Then we must have $\tau_i \notin \tau^\sreq$.
      In this case, $\nat_{S_1^\sreq}(\tau_i) = \nat_{S_1}(\tau_i),
      \nat_{S_2^\sreq}(\tau_i) = \nat_{S_2}(\tau_i)$, and since $S_1 \idle S_2 $, we know
      that $\nat_{S_2}(\tau_i) \leq \nat_{S_1}(\tau_i)$.
      Hence, $\nat_{S_2^\sreq}(\tau_i) \leq \nat_{S_1^\sreq}(\tau_i)$.

    \item Let $\tau_i$ be \sth{} $\rct_{S_1^\sreq}(\tau_i) > 0$.
      Then either $\tau_i \in \tau^\sreq$ and $\nat_{S_1^\sreq}(\tau_i) = T_i =
      \nat_{S_2^\sreq}(\tau_i)$.
      Or $\tau_i \notin \tau^\sreq$ and $\nat_{S_1^\sreq}(\tau_i) = \nat_{S_1}(\tau_i),
      \nat_{S_2^\sreq}(\tau_i) = \nat_{S_2}(\tau_i)$, and since $S_1 \idle S_2 $, we know
      that $\nat_{S_2}(\tau_i) = \nat_{S_1}(\tau_i)$, thus $\nat_{S_2^\sreq}(\tau_i) =
      \nat_{S_1^\sreq}(\tau_i)$.

  \end{enumerate}

  Let $\tau_i \in \Active{S_1^\sreq}$, hence $\rct_{S_1^\sreq}(\tau_i) > 0$.
  In this case, and since $S_1^\sreq \idle S_2^\sreq $, we conclude that
  $\nat_{S_1^\sreq}(\tau_i) = \nat_{S_2^\sreq}(\tau_i)$, $\rct_{S_1^\sreq}(\tau_i) =
  \rct_{S_2^\sreq}(\tau_i)$ and $\crit_{S_1^\sreq} = \crit_{S_2^\sreq}$.
  Thus, since $\schedule$ is deterministic by hypothesis, $\schedule(S_1^\sreq) = \schedule(S_2^\sreq)$,
  by \autoref{run}.
  Let $S_2^\srun$ be the (unique) state \sth{} $S_2^\sreq \CtTrans{\schedule(S_2^\sreq)}
  S_2^\srun$ and let us show that $S_1^\srun \idle S_2^\srun $:

  \begin{enumerate}

    \item By \autoref{texec}: $\crit_{S_1^\srun} = \crit_{S_1^\sreq}, \crit_{S_2^\srun} =
      \crit_{S_2^\sreq}$.
      Since $S_1^\sreq \idle S_2^\sreq $, we also have that $\crit_{S_2^\sreq} =
      \crit_{S_1^\sreq}$.
      We conclude that $\crit_{S_2^\srun} = \crit_{S_1^\srun}$.

    \item For $\tau_i = \schedule(S_2^\sreq) = \schedule(S_1^\sreq) : \rct_{S_1^\srun}(\tau_i) =
      \rct_{S_1^\sreq}(\tau_i) -1, \rct_{S_2^\srun}(\tau_i) = \rct_{S_2^\sreq}(\tau_i) -1$.
      Since $S_1^\sreq \idle S_2^\sreq $, we know that $\rct_{S_2^\sreq} =
      \rct_{S_1^\sreq}$.
      Hence $\rct_{S_2^\srun}(\tau_i) = \rct_{S_1^\srun}(\tau_i)$.
      For $\tau_i \neq \schedule(S_2^\sreq) = \schedule(S_1^\sreq) : \rct_{S_1^\srun}(\tau_i) =
      \rct_{S_1^\sreq}(\tau_i), \rct_{S_2^\srun}(\tau_i) = \rct_{S_2^\sreq}(\tau_i)$.
      Hence $\rct_{S_2^\srun}(\tau_i) = \rct_{S_1^\srun}(\tau_i)$.
      We then conclude that $\rct_{S_2^\srun} = \rct_{S_1^\srun}$.

    \item By \autoref{texec}: $\nat_{S_1^\srun}(\tau_i) = \max\{\nat_{S_1^\sreq}(\tau_i)-1,
      0\}$ and $\nat_{S_2^\srun}(\tau_i) = \max\{\nat_{S_2^\sreq}(\tau_i)-1, 0\}$.

      \begin{enumerate}

        \item When $\tau_i$ is \sth{} $\rct_{S_1^\srun}(\tau_i) = 0$, since $S_1 \idle S_2 $, we know that $\nat_{S_2^\sreq}(\tau_i) \leq \nat_{S_1^\sreq}(\tau_i)$.
          We conclude that $\nat_{S_2^\srun}(\tau_i) \leq \nat_{S_1^\srun}(\tau_i)$.

        \item  When $\tau_i$ is \sth{} $\rct_{S_1^\srun}(\tau_i) > 0$, we have
          $\rct_{S_1^\sreq}(\tau_i) > 0$ too as $\rct_{S_1^\sreq}(\tau_i) \geq
          \rct_{S_1^\srun}(\tau_i) > 0$.
          Since $S_1^\sreq \idle S_2^\sreq $ and $\rct_{S_1^\sreq}(\tau_i) >0$ we have also
          $\nat_{S_2^\sreq}(\tau_i) = \nat_{S_1^\sreq}(\tau_i)$.
          We conclude that $\nat_{S_2^\srun}(\tau_i) = \nat_{S_1^\srun}(\tau_i)$.

      \end{enumerate}

  \end{enumerate}

  Let $S_2^\scom$ be the (unique) state \sth{} $S_2^\srun \CpTrans{\schedule(S_2^\sreq),\theta}
  S_2^\scom$.
  Per \autoref{tcom}:
  \[S_2^\scom \in \{S_2^\srun, \sigCmp{S^\srun_2}{\tau_r},
  \critUp{S^\srun_2}{\tau_r}\}\] with $\tau_r=\schedule(S_2^\sreq)=\schedule(S_1^\sreq)$.
  Obviously, $S_1^\scom = S_1^\srun \idle S_2^\scom = S_2^\srun $ for the first outcome.
  Next, let us prove that $S_1^\scom = \sigCmp{S^\srun_1}{\tau_r} \idle S_2^\scom =
  \sigCmp{S^\srun_2}{\tau_r}$:

  \begin{enumerate}

    \item $\crit_{S_1^\scom} = \crit_{S_1^\srun}, \crit_{S_2^\scom} = \crit_{S_2^\srun}$.
      Since $S_1^\srun \idle S_2^\srun $, we also have that $\crit_{S_2^\srun} =
      \crit_{S_1^\srun}$.
      We conclude that $\crit_{S_2^\scom} = \crit_{S_1^\scom}$.
    \item $\rct_{S_2^\scom}(\tau_r) = 0 = \rct_{S_1^\scom}(\tau_r)$.
      For $\tau_i \neq \tau_r$, we have $\rct_{S_1^\scom}(\tau_i) =
      \rct_{S_1^\srun}(\tau_i)$ and $\rct_{S_2^\scom}(\tau_i) = \rct_{S_2^\srun}(\tau_i)$.
      Since $S_1^\srun \idle S_2^\srun $, we know that $\rct_{S_2^\srun}(\tau_i) =
      \rct_{S_1^\srun}(\tau_i)$.
      Hence, $\rct_{S_2^\scom}(\tau_i) = \rct_{S_1^\scom}(\tau_i)$.
      We conclude that $\rct_{S_2^\scom} = \rct_{S_1^\scom}$.

    \item For all $\tau_i \in \tau: \nat_{S_1^\scom}(\tau_i) = \nat_{S_1^\srun}(\tau_i)$
      and $\nat_{S_2^\scom}(\tau_i) = \nat_{S_2^\srun}(\tau_i)$

      \begin{enumerate}
        \item When $\tau_i$ is \sth{} $\rct_{S_1^\scom}(\tau_i) = 0$, since $S_1^\srun \idle S_2^\srun$, we know that $\nat_{S_2^\srun}(\tau_i) \leq
          \nat_{S_1^\srun}(\tau_i)$.
          Thus, $\nat_{S_2^\scom}(\tau_i) \leq \nat_{S_1^\scom}(\tau_i)$.
        \item When $\tau_i$ is \sth{} $\rct_{S_1^\scom}(\tau_i) > 0$ then $\tau_i \neq
          \tau_r$, therefore $\rct_{S_1^\scom}(\tau_i) = \rct_{S_1^\srun}(\tau_i)$ and
          since $S_1^\srun \idle S_2^\srun $, we know that $\nat_{S_2^\srun}(\tau_i) =
          \nat_{S_1^\srun}(\tau_i)$.
          Thus, $\nat_{S_2^\scom}(\tau_i) = \nat_{S_1^\scom}(\tau_i)$.

      \end{enumerate}

  \end{enumerate}

  Similarly, let us also prove that $S_1^\scom = \critUp{S^\srun_1}{\tau_r} \idle S_2^\scom
  = \critUp{S^\srun_2}{\tau_r}$:

  \begin{enumerate}

    \item $\crit_{S_2^\scom} = \hi = \crit_{S_1^\scom}$.

    \item $\rct_{S_2^\scom}(\tau_r) = C_r(\hi)-C_r(\lo) = \rct_{S_1^\scom}(\tau_r)$.
      For $\tau_i \neq \tau_r$, we have $\rct_{S_1^\scom}(\tau_i) =
      \rct_{S_1^\srun}(\tau_i)+C_i(\hi)-C_i(\lo)$ when $L_i = \hi \wedge
      \rct_{S_1^\srun}(\tau_i)>0$ and $\rct_{S_1^\scom}(\tau_i) = 0$ otherwise.
      Similarly, we have $\rct_{S_2^\scom}(\tau_i) =
      \rct_{S_2^\srun}(\tau_i)+C_i(\hi)-C_i(\lo)$ when $L_i = \hi \wedge
      \rct_{S_2^\srun}(\tau_i)>0$ and $\rct_{S_2^\scom}(\tau_i) = 0$ otherwise.
      Since $S_1^\srun \idle S_2^\srun $, we know that $\rct_{S_2^\srun}(\tau_i) =
      \rct_{S_1}^\srun(\tau_i)$.
      Thus $\rct_{S_2^\scom}(\tau_i) = \rct_{S_1^\scom}(\tau_i)$ for both cases.
      We conclude that $\rct_{S_2^\scom} = \rct_{S_1^\scom}$.

    \item For all $\tau_i \in \tau: \nat_{S_1^\scom}(\tau_i) = \nat_{S_1^\srun}(\tau_i)$
      and $\nat_{S_2^\scom}(\tau_i) = \nat_{S_2^\srun}(\tau_i)$

      \begin{enumerate}
        \item When $\tau_i$ is \sth{} $\rct_{S_1^\scom}(\tau_i) = 0$, since $S_1^\srun \idle S_2^\srun$, we know that $\nat_{S_2^\srun}(\tau_i) \leq
          \nat_{S_1^\srun}(\tau_i)$.
          Thus, $\nat_{S_2^\scom}(\tau_i) \leq \nat_{S_1^\scom}(\tau_i)$.

        \item When $\tau_i$ is \sth{} $\rct_{S_1^\scom}(\tau_i) > 0$:

          \begin{enumerate}
            \item for $\tau_i \neq \tau_r$ then $\rct_{S_1^\srun}(\tau_i) > 0$.
              Since $S_1^\srun \idle S_2^\srun $ and $rct_{S_1^\srun}(\tau_i) >0$ we have
              also $\nat_{S_2^\srun}(\tau_i) = \nat_{S_1^\srun}(\tau_i)$.
              We conclude that $\nat_{S_2^\scom}(\tau_i) = \nat_{S_1^\scom}(\tau_i)$.

            \item for $\tau_r$ then $\rct_{S_1^\srun}(\tau_r) = 0$ per definition, and
              $S_1^\srun \idle S_2^\srun $ only implies that $\nat_{S_2^\srun}(\tau_r)
              \leq \nat_{S_1^\srun}(\tau_r)$.
              But, since $\tau_r = \schedule(S_1^\sreq)$, we must have that
              $\rct_{S_1^\sreq}(\tau_r) > 0$ because $\schedule(S_1^\sreq) \in \Active{S_1}$.
              Hence $\nat_{S_2^\sreq}(\tau_r) = \nat_{S_1^\sreq}(\tau_r)$ by $S_1^\sreq \idle S_2^\sreq$.
              Therefore $\nat_{S_2^\srun}(\tau_r) = \nat_{S_1^\srun}(\tau_r)$ by
              \autoref{texec} and we conclude that $\nat_{S_2^\scom}(\tau_r) =
              \nat_{S_1^\scom}(\tau_r)$.

          \end{enumerate}

      \end{enumerate}

  \end{enumerate}

  Then observe that $S_1^\srun \idle S_2^\srun $ implies $\rct_{S_2^\srun} =
  \rct_{S_1^\srun}$ and $\crit_{S_2^\srun} = \crit_{S_1^\srun}$, thus
  $\Completed{S^\srun_2} = \Completed{S^\srun_1}$.
  Notice how the three outcomes for $S_1^\scom$ depends only on $\rct_{S_1^\srun}, \Completed{S^\srun_1}, \theta$ and $\schedule(S_1^\sreq)$.
  As $\theta$ is fixed, and as we established that $\schedule(S_1^\sreq)=\schedule(S_2^\sreq)$ and since $S_1^\srun \idle S_2^\srun$, we have $\rct_{S_1^\srun}=\rct_{S_2^\srun}$, then the outcome of the signal transition for $S_2^\srun$ will be the same as for $S_1^\srun$.
  Hence, $S_1^\scom \idle S_2^\scom $ for every case.

  To conclude the proof, it remains to prove that, if $S_1 \idle S_2 $ and $S_1 \in
  \DeadlineMiss{\tau}$ then $S_2 \in \DeadlineMiss{\tau}$ also.
  Let $\tau_i$ such that $\rct_{S_1}(\tau_i) > 0 \wedge \ttd_{S_1}(\tau_i) \leq 0$.
  Since $S_1 \idle S_2 $ : $\rct_{S_2}(\tau_i) = \rct_{S_1}(\tau_i)$ and $\nat_{S_2}(\tau_i)
  \leq \nat_{S_1}(\tau_i)$, thus $\ttd_{S_2}(\tau_i) \leq \ttd_{S_1}(\tau_i)$.
  Hence $\rct_{S_2}(\tau_i) > 0 \wedge \ttd_{S_2}(\tau_i) \leq \ttd_{S_1}(\tau_i)
  \leq 0$ and therefore $S_2 \in \DeadlineMiss{\tau}$.

\end{proof}

\end{document}